\documentclass[smallextended]{svjour3}

\usepackage{mathtools}
\usepackage{amsmath,amsfonts,amssymb,mathrsfs}
\usepackage[T1, OT1]{fontenc}
\DeclareTextSymbolDefault{\DJ}{T1}
\usepackage{booktabs}
\usepackage{tabularx}
\usepackage{float}
\usepackage{siunitx}
\usepackage{subcaption}
\usepackage{caption}

\usepackage{xcolor}
\usepackage{graphicx}
\usepackage{subcaption}
\AtBeginDocument{
\heavyrulewidth=.08em
\lightrulewidth=.05em
\cmidrulewidth=.03em
\belowrulesep=.65ex
\belowbottomsep=0pt
\aboverulesep=.4ex
\abovetopsep=0pt
\cmidrulesep=\doublerulesep
\cmidrulekern=.5em
\defaultaddspace=.5em
}
\usepackage{hyperref}
\hypersetup{unicode=true, 
  pdftoolbar=true, 
  pdfmenubar=true, 
  pdffitwindow=false, 
  pdfstartview={FitH}, 
  pdftitle={My title}, 
  pdfauthor={Author}, 
  pdfsubject={Subject}, 
  pdfcreator={Creator}, 
  pdfkeywords={keyword1, key2, key3}, 
  pdfnewwindow=true, 
  colorlinks=true, 
  linkcolor=blue, 
  citecolor=blue, 
  filecolor=blue, 
  urlcolor=blue, 
  pdfencoding=auto, psdextra }
\usepackage{booktabs}
\newtheorem{defi}{Definition}
\smartqed
\newtheorem{thm}{Theorem}
\newtheorem{lem}[thm]{Lemma}
\usepackage{tabstackengine}
\stackMath%
\usepackage{algpseudocode}  

\algrenewcommand{\algorithmiccomment}[1]{\hskip2em\textcolor{black}{//\scriptsize\tt #1}}

\usepackage{newfloat}
\usepackage{caption}
\AtEndEnvironment{algorithm}{\vskip-6pt\noindent\rule{\linewidth}{1pt}\par\nobreak\vskip-8pt}

\DeclareFloatingEnvironment[
    fileext=loa,
    listname=List of Algorithms,
    name=ALGORITHM,
    placement=htbp,
]{algorithm}
\DeclareCaptionFormat{algorithms}{\vskip-4pt\rule{\linewidth}{1pt}\par\centering#1#2#3\vskip-5pt\rule{\linewidth}{0.6pt}\vskip-6pt}
\algblock[Input]{Input}{EndInput}
\algblockdefx[Input]{Input}{EndInput}%
    [1]{\textbf{Input} #1}%
    {}
\newcommand{\real}{\mathbb{R}} \newcommand{\complex}{\mathbb{C}} 
\newcommand{\rank}[1]{\mathrm{rank} (#1)} \newcommand{\range}[1]{\mathcal{R} (#1)}
\newcommand{\norm}[1]{\left\lVert#1\right\rVert} 
\DeclareMathOperator{\trace}{Tr}

\DeclareMathOperator*\lowlim{\underline{lim}}

\usepackage{mathtools}

\def\symbtype{0} \ifcase\symbtype%
\newcommand{\adj}[1]{{#1}^{*}}

\DeclarePairedDelimiter\bra{\langle}{\rvert}
\DeclarePairedDelimiter\ket{\lvert}{\rangle}
\DeclarePairedDelimiterX\braket[2]{}{}{(#1, #2)}
\DeclarePairedDelimiterX\mket[2]{}{}{#1\otimes#2}

\or%
\newcommand{\adj}[1]{{#1}^{\dagger}}

\DeclarePairedDelimiter\bra{\langle}{\rvert}
\DeclarePairedDelimiter\ket{\lvert}{\rangle}
\DeclarePairedDelimiterX\mket[2]{\lvert}{\rangle}{#1,#2}
\DeclarePairedDelimiterX\braket[2]{\langle}{\rangle}{#1 \delimsize\vert#2}

\fi

\usepackage{cancel}


\renewcommand{\intercal}{\mathsf{T}}
\usepackage{dsfont}
\newcommand{\I}{\mathds 1}
\newcommand{ \trans}[1]{{#1}^{\intercal}}
\newcommand{\myH}[1][{}]{\mathcal{H}_{#1}}
\newcommand{\etal}{{\sl et~al.}}
\newcommand\symmg{\mathcal{S}}
\newcommand{\symmf}{\mathcal{G}}
\newcommand{\abs}[1]{\left|#1\right|}
\newcommand{\sqp}{\text{\tiny SQP}}
\newcommand{\nt}{\text{\tiny NT}}
\newcommand{\Lag}{\mathcal{L}}

\begin{document}
\title{Decomposition of completely symmetric states}

\author{Qian Lilong\and Chu Delin}

\institute{L. Qian
  \at Department of Mathematics, National University of Singapore, Singapore \\
  \email{qian.lilong@u.nus.edu} 
  \and D. Chu
  \at 
  Department of Mathematics, National University of Singapore, Singapore \\\email{matchudl@nus.edu.sg} }
\date{\today}
\maketitle
\begin{abstract}
	Symmetry is a fundamental milestone of quantum physics and the relation between entanglement  is one of the central mysteries of quantum mechanics. In this paper, we consider a subclass of symmetric  quantum states in the bipartite system, namely, the completely symmetric 
 states, which is invariant under the index permutation.  We  investigate the separability of these states.
 After studying some examples, we conjecture that completely symmetric state is  separable if and only if it is 
  S-separable, i.e., each term in this decomposition is a    symmetric pure product
 state $\ket{x,x}\bra{x,x}$.  It was proved to be true when the rank is
 less than or equal to $4$ or $N+1$.  
   After studying the properties of these state, we 
 propose a numerical algorithm which is able to detect S-separability. This algorithm is based on the best separable approximation,  which furthermore turns out to be applicable to test the
 separability of quantum states in bosonic system.
 Besides, we analysis  the convergence behaviour
 of this algorithm. Some numerical examples are tested to show the effectiveness of the algorithm.
\end{abstract}


\section{Introduction}
\label{intro}
Quantum entanglement, first recognized by Eistein~\cite{einstein1935can} and
Schr\"odinger~\cite{schrodinger1935gegenwartige}, plays  a crucial role in  the field of quantum computation and quantum
information. It is  the resources of many applications  in  quantum 
cryptography, quantum teleportation, and quantum key distribution\cite{nielsen2002quantum}. Therefore, the question of
whether a given quantum state is entangled or separable is both of fundamental importance. For a given quantum state
$\rho$ acting on the finite-dimensional Hilbert space $\mathcal{H}_1\otimes \mathcal{H}_{2}$, it is said to be separable if
it can be written as a convex linear combination of pure product quantum states, i.e.,
\begin{equation}
  \label{goesqp:eq:83}
  \rho = \sum_i \lambda_i\rho_{i}^{(1)}\otimes \rho_{i}^{(2)},
\end{equation}
where $\sum_i\lambda_i=1,\lambda_i>0$ and $\rho_i^{(1)}$ and $\rho_i^{(2)}$ are the pure  states in the subspaces $\mathcal{H}_{1}$ and
$\mathcal{H}_2$, respectively.

Despite its wide importance,  to find an efficient and effective  method to solve this question  in general case  is
considered to be NP-hard~\cite{gurvits2003classical,gharibian2008strong}.  After remarkable efforts  over recent
years,  many  methods to test entanglement  are proposed for some  subclasses of quantum states. For example,
one of the most famous criteria is positive-partial-transpose (PPT)~\cite{Peres1996a}. It tells that if a quantum state
is separable then its partial transposed state is necessarily positive, which is called PPT state. Moreover, the PPT
criterion is also sufficient for the case 
$\mathrm{dim}(\mathcal{H}_{1})\cdot\mathrm{dim}(\mathcal{H}_{2})\leqslant6$~\cite{Horodecki1996,Woronowikz1976}. A
natural generalization of PPT criterion is the permutation criterion~\cite{Horodecki2006}, namely, whatever permutation
of indices of a state written in product basis leads to a separability criterion. Besides, a subclass of PPT states,
strong PPT states, are considered in Refs.~\cite{chruscinski2008quantum,ha2010entangled,Qian_2018}. The strong PPT
states are proved to 
be separable when $\mathrm{dim}(\myH[1])=2$ and $\mathrm{dim}(\myH[2])\leqslant4$. The low-rank
quantum states are also investigated.  In the the bipartite system, it was proved  that the state $\rho$ is separable if
$\rank{\rho} = \mathrm{dim}(\myH[2])$~~\cite{kraus2000separability,horodecki2000operational}. See
Ref~\cite{fei2003separability,Li_2014}
for generalized results in the multipartite system.

Some numerical methods to test the entanglement are also suggested. Doherty {\sl et al.}~\cite{doherty2004complete}
introduced an iterative algorithm which is based on symmetric extension. That is, if a state $\rho$ is
separable on $\myH[1] \otimes \myH[2]$, then it must have a symmetric extension on $\myH[1]\otimes\myH[2]\otimes\myH[1]$.
It gives a hierarchy condition for separability  and can be applied recursively if the extension is found. During the
process,  each test is  at least as powerful as all the previous ones, where the  first one is equivalent to the  PPT
criterion. If the state is entangled, then this algorithm will terminate after finite steps. But, if it is
separable,  this algorithm will never stop. 
Another numerical algorithm, analytical cutting-plane method, was proposed by Ioannou \etal~\cite{ioannou2004improved},
which is based on the entanglement witness. It finds the entanglement witness by reducing the set of traceless operators
step by step. 
Dahl \etal \cite{dahl2007tensor} also proposed a feasible descent  method to find the closest separable state from the
aspect of
geometric structure. They utilized the  Frank-Wolfe method~\cite{bertsekas1999nonlinear} to solve  this problem.

Due to the close relation between symmetric and entanglement, it is of great interest to study the entanglement of 
the “symmetric” states. 
In this paper, we focus on a subclass of quantum states, completely symmetric  states, i.e.,  invariant under any index permutation. The conception  of completely symmetric is inspired by the “supersymmetric tensor” in the field of tensor decomposition~\cite{kofidis2002best,qi2005eigenvalues,kolda2009tensor}. In fact, any multipartite state can be regarded as a $4$-th order symmetric. Then the state is separable if it admit a special kind of decomposition.
It is thus hopeful to borrow some powerful technologies in the field of tensor analysis to tackle the entanglement problem.
Here, we investigate their properties on the bipartite system. It is expected to find more special structure with respect to the completely symmetric property.
Moreover, it is believed that the completely symmetric state is separable iff it is S-separable, which is a 
conjecture we made in this paper. In fact, we proved that this conjecture holds for
the states whose ranks are at most $4$ or $N+1$. In addition, we propose a numerical method to
test the S-separability, with the similar idea form~\cite{dahl2007tensor}. This algorithm can be used to  find the closest
S-separable states, where the distance between them can be regarded as a measurement of the entanglement. During the algorithm, we need to solve a sub-problem at each iteration, i.e., finding the maximizer
of the following optimization problem
\begin{equation}
  \label{S-sep-v7:eq:2}
  \begin{array}{rl}
    \max\limits_{x} & \bra{x,x} \rho - \rho_k\ket{x,x}\\
    \text{s.t.} &  \norm{x}=1,
  \end{array}
\end{equation}
where $\rho_k$ is the S-separable state at $k$-th step to approximate the closest S-separable state of $\rho$.
 For this sub-problem,  we suggest a  sequential quadratic programming (SQP)
algorithm. Numerical experiments are tested to show the efficiency of this method. It is hopeful that our study would
shed new lights on understanding the structure of the entangled states.

This paper is organized as follows. In the next section, we summarize the necessary preliminaries. In
Sec.~\ref{sec:prop-symm-separ}, we investigate the properties of  S-separable states and study theoretically the problem
 that whether completely symmetric states admit the S-decomposition. In Sec.~\ref{sec:real-super-symmetric}, we propose a feasible
  descent method to check S-separability. In Sec.~\ref{sec:algor-optim-probl}, we
  suggest an SQP algorithm for the subproblem, which is the key step in the feasible descent method. That is to find
  maximizer of the optimization problem~\eqref{S-sep-v7:eq:2}.
  In Sec.~\ref{converge}, we study the convergence behavior of the algorithm proposed in the previous section.  In
  Sec.~\ref{numerical}, the 
  numerical examples are tested to show 
  the effectiveness of the proposed algorithm. Some suggestions for improving the numerical method are discussed in Sec.~\ref{sec:improve}. Finally, we give some concluding remarks in Sec.~\ref{remarks}.
 
  \section{Preliminaries}%
  \label{sec:pre}
  In this paper, we consider the quantum states in  the system $\myH[1]\otimes\myH[2]$, where $\myH[1]$ and
  $\myH[2]$ are two Hilbert spaces.  Mathematically, any quantum states can be represented by the Hermitian positive matrices
  with trace one, which are called the density matrices. For pure state $\rho$, its density matrix  is of rank one,
  \begin{equation}
    \label{goesqp:eq:84}
    \rho = \ket{\phi}\bra{\phi},
  \end{equation}
  where $\ket{\phi}$ is the vector in the  space $\myH[1]\otimes \myH[2]$. In particular, $\rho$ is said to be a pure product 
  state if
  \begin{equation}
    \label{goesqp:eq:85}
    \ket{\phi} = \ket{x}\ket{y},
  \end{equation}
  where $\ket{x}$ and $\ket{y}$ are vectors in the spaces $\myH[1]$ and $\myH[2]$, respectively.

  In this paper, we are interested in the quantum states which are invariant under any index permutation. In this case,
  $\myH[1]=\myH[2]$, denoted by 
  $\myH$.  Here 
   the definition of completely symmetric states is given formally as follows.

  \newcommand{\bH}{\myH\otimes \myH} Suppose $\rho$ is a quantum state in the bipartite system $\bH$ and
  $\mathrm{dim}(H)=N$, then it can be written as
  \begin{equation}
    \label{symmetrictensor:eq:1}
    \rho = \sum_{i,j,k,l=0}^{N-1}\rho_{ijkl}\ket{i,k}\bra{j,l},
  \end{equation}
  where $\ket{0},\ket{1},\ldots,\ket{N-1}$ is a natural basis in the space $\myH$.

\begin{defi}
  Let $\rho$ be a quantum state as in Eq.\eqref{symmetrictensor:eq:1}. Then $\rho$ is said to be completely symmetric if
  \begin{equation}
    \label{symmetrictensor:eq:9}
    \rho_{\pi(ijkl)} = \rho_{ijkl},
  \end{equation}
  for any index permutation $\pi$.
\end{defi}
On the other hand,  $\rho$ can be written as a block matrix:
\begin{equation}
  \label{goesqp:eq:90}
  \rho =
  \begin{pmatrix}
    \rho_{00}&\cdots&\rho_{1,N-1}\\
    \vdots&\ddots&\vdots\\
    \rho_{N-1,0}&\cdots&\rho_{N-1,N-1}
  \end{pmatrix},
\end{equation}
where each block is an $N\times N $ matrix:
\begin{equation}
  \label{goesqp:eq:91}
  \rho_{ij} =
  \begin{pmatrix}
    \rho_{ij00}&\cdots&\rho_{ij0,N-1}\\
    \vdots&\ddots&\vdots\\
    \rho_{ij,N-1,N-1}&\cdots&\rho_{ij,N-1,N-1}
  \end{pmatrix}.
\end{equation}
Hence, we have $\rho_{ijkl} = (\rho_{ij})_{kl},\forall i,j,k,l<N$. The coefficients $\rho_{ijkl}$ thus
correspond to the density matrix $\rho$ in  an intuitive way.  
Note that, in this paper, we index from 0 instead of 1
in order to be consistent with the standard notations in quantum computation. Since any quantum state is a Hermitian
matrix, we have $\rho=\rho^{\dagger}$. By 
Eq. \eqref{symmetrictensor:eq:9}, we have $\rho=\rho^{\intercal}$ for the completely symmetric state, where $\intercal$ is the
partial transpose operator. It follows that $\rho$ is a
real matrix. Therefore,  hereafter in this paper, we
consider in details only the real case, i.e., $\myH=\real^{N}$.

Let us, for the moment, restrict ourselves to supported  states, i.e., the states are supported on
$\bH$~\cite{kraus2000separability}. For completely symmetric states, it is easy to check that their reduced states are
identical, that is,
\begin{equation}
  \label{goesqp:eq:86}
  \rho_1 = \rho_2,
\end{equation}
where
\begin{equation*}
  \label{goesqp:eq:87}
  \begin{split}
    \rho_1 &= (\I\otimes \trace)\rho\\
    &=  \sum_{i,j=0}^{N-1}\left(\sum_{k=0}^{N-1}\rho_{ijkk}\right)\ket{i}\bra{j} \\
    &=
    \begin{pmatrix}
       \trace\rho_{ij00}&\cdots&\trace\rho_{ij0,N-1}\\
    \vdots&\ddots&\vdots\\
   \trace \rho_{ij,N-1,N-1}&\cdots&\trace\rho_{ij,N-1,N-1}
    \end{pmatrix},\\
    \rho_2 &= (\trace\otimes\I)\rho\\
    &= \sum_{k,l=0}^{N-1}\left(\sum_{i=0}^{N-1}\rho_{iikl}\right)\ket{k}\bra{l}\\
    & = \sum_{i=0}^{N-1}\rho_{ii}.
  \end{split}
\end{equation*}

Then $\rho$ is said to be supported on $\bH$ iff $\range{\rho_1} = \myH$. Note that for completely symmetric state,
$\rho^{\intercal_1} = (\intercal\otimes \I) = \rho$,  $\rho$ is then necessarily PPT.\@

A state is said to be separable if it can be written as the convex combination of pure product states as in
Eq. \eqref{goesqp:eq:83}. Compared with this general case, completely symmetric states may have some better properties. Now we
introduce a subclass of completely symmetric states:
\begin{defi}
	Let $\rho=\ket{\phi}\bra{\phi} $ be a pure state in $\bH$. $\rho$ is said to be a completely symmetric pure product state if 
	\begin{equation}
		\phi = \alpha \ket{x,x},\alpha\in\real,x\in\myH.
	\end{equation}
\end{defi}
For the mixed states, we have:
\begin{defi}
   The quantum state $\rho$ in the   system  $\bH$ is said to be symmetrically separable (S-separable) if it can be written as a 
  convex combination of completely symmetric pure product state:
  \begin{equation}
    \label{goesqp:eq:63}
    \rho = \sum_{i=0}^{L-1} p_i \ket{x_i,x_i}\bra{x_i,x_i},\;x_i\in\myH.
  \end{equation}
\end{defi}
Here  the decomposition of Eq. \eqref{goesqp:eq:63} is called
S-decomposition and $L$ is called the length of the decomposition.

Let $\symmf$ denote the space of completely symmetric matrices (not necessarily positive) and let $\symmg$ denote the convex set
generated by the completely symmetric pure product states. Then $\symmg$ is a compact convex subset of $\symmf$. 
The space $\symmf$ is equipped with the standard inner product
\begin{equation}
  \label{goesqp:eq:88}
  \langle\rho,\sigma\rangle = \trace(\rho\sigma ),\rho,\sigma\in\symmf.
\end{equation}
The associated matrix norm is the Frobenius norm
\begin{equation}
  \label{goesqp:eq:89}
  \norm{\rho}_{F} ={\langle\rho,\rho\rangle}^{\frac{1}{2}} =  {(\trace(\rho^2))}^{\frac{1}{2}}.
\end{equation}
For this norm, we have
\begin{equation}
  \label{goesqp:eq:118}
  \norm{\rho}_{F}\leqslant \norm{\rho^{\frac{1}{2}}}_{F} = \sqrt{\trace{\rho}},
\end{equation}
which implies that any   quantum state has norm less than or equal to 1. Besides, the spectral norm
$\norm{\rho}_{2}$ are 
also used for matrices in this paper, which is the largest singular value of $\rho$. In this paper, for vectors, we use $l_2$
norm, denoted by $\norm{x}$.

\section{Properties of completely symmetric states}%
\label{sec:prop-symm-separ}
In this section, consider the completely symmetric state $\rho$ in the system  $\bH$, where $\myH=\real^{N}$. We investigate 
the properties of the completely symmetric states and S-separable states. We prove that any 
rank 1, 2, and 3 completely symmetric states are S-separable. Moreover, any rank $N$ completely symmetric  state supported on the
$\bH$ space is S-separable.  Forward,  stronger 
results will be obtained    for separable completely symmetric states. Therefore, a natural question  arises, during the
analysis, whether the completely symmetric 
states  are S-separable.

To begin with, we show that the completely symmetric and S-separable properties  are invariant  under the invertible local operator
$A\otimes A,A\in \real^{N\times N}$.
\begin{lem}
  Suppose $\rho$ is a quantum state in the system $\bH$ and $A$ is an invertible operator on $\myH$. Then
  \begin{enumerate}
    \item $\rho$ is completely symmetric if and only if  $(A\otimes A)\rho (A\times A)^{\intercal}$ is completely symmetric.
          
    \item $\rho$ is S-separable if and only if  $(A\otimes
          A)\rho (A\otimes A)^{\intercal}$ is S-separable.    
  \end{enumerate}
  \end{lem}
  \begin{proof}
    It suffice to prove only one side for the two conclusions. Otherwise, we can consider the state which is  applied by the
    invertible local operator $(A^{-1}\otimes A^{-1})$.
    \begin{enumerate}
      \item  Suppose $\rho$ is a completely symmetric states as in Eq.\eqref{symmetrictensor:eq:1}.
  Let $\sigma = (A\otimes A)\rho (A\otimes A)^{\intercal}$, which can  be written as 
  \begin{equation}
    \label{goesqp:eq:66}  
    \sigma  = \sum_{i',j',k',l'=0}^{N-1} \sigma_{i'j'k'l'}\ket{i',k'}\bra{j',l'}.
  \end{equation}
  Let
  \begin{equation}
    \label{goesqp:eq:67}
    A =
    \begin{pmatrix}
      a_{00} & \cdots & a_{0,N-1}\\
      \vdots & \ddots & \vdots\\
      a_{N-1,0} & \cdots & a_{N-1,N-1}
    \end{pmatrix}.
  \end{equation}
  Hence,
  \begin{equation}
    \label{goesqp:eq:68}
    A\ket{i} = \sum_{i'=0}^{N-1}a_{i'i}\ket{i'}.
  \end{equation}
            Therefore,
  \begin{equation}
    \label{goesqp:eq:69}
    \sigma = \sum_{\substack{i,j,k,l=0\\i',j',k',l'=0}}^{N-1}\!\!\!\!\!\!\rho_{ijkl}a_{i'i}a_{j'j}a_{k'k}a_{l'l}\ket{i',k'}\bra{j',l'}.
  \end{equation}
  In order to prove $\sigma$ is  completely symmetric, we only need to prove $\sigma_{i'j'k'l'} = \sigma_{\pi(i'j'k'l')}$
            where $\pi$ is an arbitrary index  permutation.
            
  Note that
  \begin{equation}
    \label{goesqp:eq:70}
    \sigma_{i'j'k'l'}  = \sum_{i,j,k,l}\rho_{ijkl}a_{i'i}a_{j'j}a_{k'k}a_{l'l}.
  \end{equation}
  Then
  \begin{equation}
    \label{goesqp:eq:71}
    \begin{split}
      \sigma_{j'i'k'l'} &= \sum_{i,j,k,l}\rho_{ijkl}a_{j'i}a_{i'j}a_{k'k}a_{l'l}\\
      & = \sum_{i,j,k,l}\rho_{jikl}a_{j'i}a_{i'j}a_{k'k}a_{l'l}.\\   
      \end{split}
    \end{equation}
            Therefore, $   \sigma_{j'i'k'l'} = \sigma_{i'j'k'l'}$. Similarly, 
            \[\sigma_{i'j'k'l'} = \sigma_{\pi(i'j'k'l')}, \]
            for any index permutation $\pi$.
      \item If $\rho$ is S-separable, then it can be written as
            \begin{equation}
              \label{goesqp:eq:20}
              \rho = \sum_i p_i\ket{x_i,x_i}\bra{x_i,x_i}.
            \end{equation}
            Let $\sigma = (A\otimes A)\rho(A\otimes A)^{\intercal}$, then we have
            \begin{equation}
              \label{goesqp:eq:21}
              \sigma = \sum_ip_i\ket{Ax_i,Ax_i}\bra{Ax_i,Ax_i}.
            \end{equation}
            Hence, $\sigma $ is also S-separable.
    \end{enumerate}
  \end{proof}
In this paper, the states we considered are, in fact, real states. According to Proposition 13 in Ref.~\cite{Chen_2013Dim}, we have the following lemma, which enable us to consider the separability in real case.
\begin{lem}
	\label{real} 
	Suppose $\rho$ is separable over the $\complex$ and 
	\begin{equation}
	\label{ginvariant}
		\rho = (\intercal \otimes \I)\rho.
	\end{equation}
then $\rho$ is separable over $\real$, that is, $\rho$ can be written as a sum of real pure product states.
\end{lem}
The states which satisfy Eq. \eqref{ginvariant} are said to be G-invariant. 
 
We begin with the simplest case where $\rho$ is a pure state.
\begin{lem}%
  \label{rank1}
Any completely symmetric pure states are S-separable. 
\end{lem}
\begin{proof}
  Note that any quantum states are assumed to be positive. Suppose that $\rho$ is a  completely symmetric pure state in the
  $N\otimes N $ system.  Hence, it can be written as
  \begin{equation}
    \label{goesqp:eq:64}
    \rho =  \ket{x,y}\bra{x,y},
  \end{equation}
  where $x,y$ are unit vectors in $\myH$.
  Therefore, there exists a unitary  operator $A$ such that $\ket{Ax} = \ket{0}$.
  We  can thus  assume that
  \begin{equation}
    \label{goesqp:eq:72}
    \begin{split}
      \rho &=  \ket{0}\bra{0}\otimes \ket{y}\bra{y}\\
      & = 
      \begin{pmatrix}
        yy^{\intercal} &\cdots & 0\\
        \vdots & \ddots & \vdots \\
        0& \cdots & 0
      \end{pmatrix}\\
      & =
      \begin{pmatrix}
        \rho_{00}   &\cdots&\rho_{0,N-1}\\
        \vdots     &\ddots&\vdots\\
        \rho_{N-1,0}&\cdots&\rho_{N-1,N-1}
      \end{pmatrix},
      \end{split}
    \end{equation}
    where
    \begin{equation*}
      \rho_{00} = yy^{\intercal},\rho_{ij} = 0,\forall i,j\neq 0.
    \end{equation*}
    By the  symmetry of $\rho$ and Eq.~\eqref{goesqp:eq:91}, we have
    \begin{equation}
      \label{goesqp:eq:76}
      \rho_{00kl}= \rho_{kl00} = 0,\;\forall k,l\neq 0.
    \end{equation}
    Let $y = (y_0,y_1,\cdots,y_{N-1})^{\intercal}$, then
    \begin{equation}
      \label{goesqp:eq:78}
      \rho_{00kk} = y_k^2 = 0,\;\forall k\neq 0.
    \end{equation}
    It follows that $\ket{y} = \pm\ket{0}$, which implies that $\rho=\ket{0,0}\bra{0,0}$. Hence, 
    $\rho$ is  S-separable.
  \end{proof}

  We can also prove that rank $N$ states  supported in the $N\otimes N$ space are S-separable.
  Before proving this result, we need a lemma to prove that $\rho$ can be written as a sum of $N$ real pure product
  states.
  \begin{lem}
    \label{nsum}
    Suppose $\rho$ is supported on $\bH$, where $\myH = \real^N$. If $\rank{\rho}=N$, then $\rho$ is a sum of $N$
    real pure product states. 
  \end{lem}
  \begin{proof}
    By the result in Ref.\cite{horodecki2000operational}, $\rho$ is a sum of $N$ pure product states:
    \begin{equation}
      \label{S-sep-v5.1:eq:14}
      \rho = \sum_{i=0}^{N-1}\ket{x_i,y_i}\bra{x_i,y_i}.
    \end{equation}
  
    Since $\rho$ is supported on $\bH$, then $\ket{x_i} (\ket{y_i}),i=0,1,\ldots,N-1$ are linearly independent, which implies that 
   any vector
    \begin{equation}
      \label{S-sep-v5.1:eq:15}
      v = \sum_{i=0}^{N-1}\lambda_i\ket{x_i,y_i},v\neq\lambda_i \ket{x_i,y_i}
    \end{equation}
    is entangled. That is $\range{\rho}$   only contains $N$  product vectors: $\ket{x_i,y_i},i=0,1,\ldots,\linebreak N-1$.

    On the other hand, by Lemma~\ref{real}, $\rho$ is separable over $\real$. Hence, there exists a real  product vector $\ket{x,y}$ in the
    range of $\rho$. Therefore, $\ket{x,y}=\alpha_i\ket{x_i,y_i},\alpha_i\in\complex$ for some $i$. Apply the same
    discussion recursively on the state
    $ \rho - \ket{x_i,y_i}\bra{x_i,y_i} = \rho - \abs{\alpha_i}^2\ket{x,y}\bra{x,y}$, which is supported on $(N-1)\otimes
    (N-1)$ space, we conclude  that $\rho$ can be written as a sum of $N$ real pure product states.
  \end{proof}
  \begin{lem}
    \label{rankN}
  Suppose $\rho$ is a completely symmetric state supported  in $N\otimes N$ space. If $\rank{\rho}= N$, then it is
  S-separable.
\end{lem}
\begin{proof}
  According to Lemma~\ref{nsum}, $\rho$ can be written as a sum of $N$ real pure product states:
  \begin{equation}
    \label{goesqp:eq:73}
    \rho = \sum_{i=0}^{N-1}p_i \ket{x_i}\bra{x_i} \otimes \ket{y_i}\bra{y_{i}},\; p_i>0,x_i,y_i\in\myH.
  \end{equation}
  And $\{x_i\}_{i=0}^{N-1}$ is set of basis vectors in $\myH$. Hence, there exists an invertible operator $A$ such that
  $A\ket{x_{i}} = \ket{i}$.
  Apply the invertible local operator $A\otimes A$ to $\rho$,  we can assume that
  \begin{equation}
    \label{goesqp:eq:74}
    \begin{split}
      \rho &= \sum_{i=0}^{N-1}p_i \ket{i}\bra{i}\otimes \ket{y_i}\bra{y_i}\\
      &=
      \begin{pmatrix}
        p_0y_0y_0^{\intercal}&0 &\cdots&0\\
        0&p_1y_1y_1^{\intercal}&\cdots&0\\
        \vdots&\vdots&\ddots&\vdots\\
        0&0&\cdots&p_{N-1}y_{N-1}y_{N-1}^{\intercal}
      \end{pmatrix},
      \end{split}
    \end{equation}
    Therefore, we have
  \begin{equation}
    \label{goesqp:eq:80}
    \rho_{ijkl} = 0,\text{ if }i\neq j.
  \end{equation}
  Let
  \begin{equation}
    \label{goesqp:eq:81}
    y_i =
    \begin{pmatrix}
      y_{i0}&
      y_{i1}&
      \cdots&
      y_{i,N-1}
    \end{pmatrix}^{\intercal}.
  \end{equation}
  By the  super symmetry of $\rho$ , we have
  \begin{equation}
    \label{goesqp:eq:82}
    \rho_{kkjj} =  p_k y_{kj}y_{kj} = \rho_{kjkj} = 0, \text{ if } k\neq j.
  \end{equation}
  It follows that
  \begin{equation}
    \label{goesqp:eq:75}
    \ket{y_{i}} = \alpha_{i}\ket{i},\,\alpha_i\in\real,\,\forall i = 0,1,\ldots,N-1.
  \end{equation}
  Hence, $\rho = \sum_{i=0}^{N-1}p_i\alpha_i^2\ket{i,i}\bra{i,i}$, which completes our proof.
\end{proof}
Note that for the completely symmetric state $\rho$, the two reduced states  $\rho_1=(\I\otimes \mathrm{Tr})\rho$ and
$\rho_{2}=(\mathrm{Tr}\otimes 
\I)\rho$ are  identical. It follows that
$\rho$ has  identical local ranges. Hence, any rank 2 completely symmetric states must be supported on $2\otimes 2$
subspace. Utilizing the above lemma, we
have the following corollary.
\begin{corollary}%
\label{rank2}
  Any rank 2 completely symmetric states are S-separable.
\end{corollary}
Before considering the rank $3$ states, we introduce a useful lemma.
\begin{lem}
  \label{lem:Ginv}
	Let $\rho$ be a G-invariant quantum state in the $2\otimes 2$ system. Then $\rho$ can be written as 
a sum of 	$\rank{\rho}$ pure product states. 
\end{lem}
\begin{proof}
	Note that $\rho$ is PPT if $\rho = (\intercal \otimes \I)\rho$. And it was proved that any PPT state in the $2\otimes 2$ system is separable~\cite{Peres1996a}. By Lemma~\ref{real}, $\rho$ can be written as a sum of real pure product states.
	Then there exists a product vector  $\ket{x,y}\in\bH$ and a real number $\lambda$ such that 
	\begin{equation}
		\rho - \lambda\ket{x,y}\bra{x,y} >0
	\end{equation}
	and 
	\begin{equation}
		\rank{\rho-\lambda\ket{x,y}\bra{x,y}}=\rank{\rho}-1.
	\end{equation}
	Note that $\rho-\lambda\ket{x,y}\bra{x,y}$ is still G-invariant for real vectors $\ket{x},\ket{y}$, hence it remains separable.
	Repeat the step above, $\rho$ thus  can be written as a sum of $\rank{\rho}$ real pure product states. 
      \end{proof}
      With the above lemma, we are ready to prove that any $2\otimes 2$ completely symmetric states are S-separable.
      \begin{lem}
   \label{product22}
  Suppose  $\rho$ is a completely symmetric state  in the $2\otimes 2$ system, then there exists a vector $\ket{x}\in\myH$  such
  that
  \begin{equation}
    \label{goesqp:eq:65}
    \ket{x,x}\in\range{\rho}.
  \end{equation}
\end{lem}
\begin{proof}
  We discuss this question with respect to the rank of $\rho$.

  If $\rank{\rho}\leqslant 2$,  by Lemma~\ref{rank1} and Corollary \ref{rank2}, 
  $\rho$ is S-separable. Then each term  in the S-decomposition of $\rho$ will satisfy Eq.~\eqref{goesqp:eq:65}.

  If $\rank\rho =3$, then by Lemma~\ref{lem:Ginv}, we have
  \begin{equation}
    \label{goesqp:eq:92}
    \rho = \sum_{i=0}^2\ket{x_i}\bra{x_{i}}\otimes\ket{y_i}\bra{y_i},
  \end{equation}
  where $x_i,y_{i}$ are vectors in $\myH$, which may not be unit.

  Consider the case where $\ket{x_i},i=0,1,2$ are affinely dependent. That is, there exist $i_{0}\neq j_{0}$, such that
  \begin{equation}
    \label{goesqp:eq:93}
    \ket{x_{i_0}} =\alpha \ket{x_{j_{0}}},\alpha\in\real,\,\alpha\neq 0.
  \end{equation}
  Assume $i_{0}=0,j_{0}=1$ for simplicity. Since $\rank{\rho}=3$,  $\ket{x_i,y_i},i=0,1,2$ are linearly
  independent. Furthermore, $\ket{y_0},\ket{y_1}$ should be also linearly independent. Hence, there exists  $\beta,\gamma\in\real$ such
  that
  \begin{equation}
    \label{goesqp:eq:94}
    \ket{x_0}=\beta\ket{y_0}+\gamma\ket{y_1}.
  \end{equation}
  Forward,
  \begin{equation}
    \label{goesqp:eq:95}
    \begin{split}
      \ket{x_0,x_0} & =  \beta \ket{x_0,y_0} + \frac{\gamma}{\alpha}\ket{x_1,y_1}\\
      & = \beta \ket{x_0,y_0} + \frac{\gamma}{\alpha}\alpha\ket{x_0,y_1}\\
      & = \ket{x_0}({\beta\ket{y_0}+\gamma\ket{y_1}}).
      \end{split}
    \end{equation}
    This implies that $\range{\rho}$ contains the completely symmetric product vector $\ket{x_0,x_0}$.

    Next, we consider the case where $\ket{x_i} (\ket{y_i}),i=0,1,2$ are affinely independent. Then there exist an invertible operator $A$
    such that
    \begin{equation}
      \label{goesqp:eq:96}
      A\ket{x_i} = \ket{i},i=0,1,
    \end{equation}
    and
    \begin{equation}
      \label{goesqp:eq:97}
      Ax_2 =
      \begin{pmatrix}
        x_{20}\\x_{21}
      \end{pmatrix}.
    \end{equation}
    Since $A\ket{x_i},i=0,1,2$ are also affinely independent, $x_{2j}\neq 0,j=0,1$. Absorb the coefficient $x_{20}$ in
    $\ket{y_2}$, we can assume that $\ket{x_i}=\ket{i},i=0,1$, and $\ket{x_2} = (1,d)^{\intercal},d\neq 0$. Apply
    another local invertible operator $B\otimes B$ in $\rho$, where
    \begin{equation}
      \label{goesqp:eq:99}
      B =
      \begin{pmatrix}
        1&0\\
        0&\frac{1}{d}
      \end{pmatrix}.
    \end{equation}
    We can further assume that
    \begin{equation}
      \label{goesqp:eq:100}
      \rho = \ket{0,y_0}\bra{0,y_0} + \ket{1,y_1}\bra{1,y_1}+\ket{x_2,y_2}\bra{x_2,y_2},
    \end{equation}
    where
    \begin{equation*}
      \ket{x_2 } =
      \begin{pmatrix}
        1&1
      \end{pmatrix}^{\intercal}.
    \end{equation*}
    Suppose $y_0 = (a_0,a_1)^{\intercal}$, $y_1 = (b_0,b_1)^{\intercal}$, and $y_2 = (c_0,c_1)^{\intercal}$. If $a_1=0$, then $\ket{0,0}$ is in the range of
    $\rho$. Hence, we assume that $a_1\neq 0$. Since $\frac{1}{a_1}\rho$ is also completely symmetric, we further assume that $a_1=1$.
    According to Eq.~\eqref{goesqp:eq:91} and by the symmetry of $\rho$, we have
    \begin{equation}
      \label{goesqp:eq:101}
      \left\lbrace
      \begin{split}
        c_0^2 = & a_0a_1+c_0c_1,\\
        c_1^2 = & b_0b_1 + c_0c_1,\\
        c_0c_1 = & 1+c_1^2=b_0^2+c_0^2.
      \end{split}\right.
    \end{equation}
    We can further  assume that $c_1>0$ otherwise we can replace $\ket{y_2}$ with $-\ket{y_2}$.
    Hence, by solving Eq.~\eqref{goesqp:eq:101}, we have
    \begin{equation}
      \label{goesqp:eq:102}
      \begin{split}
        a_0 & = 1 + \frac{1}{c_1^2},\\
        b_0 &= \sqrt{1+\frac{1}{c_1^2}},\\
        b_1 & = \frac{1}{\sqrt{1+\frac{1}{c_1^2}}},\\
        c_0 & = c_1 + \frac{1}{c_1}.
        \end{split}
      \end{equation}
      Here we assume that $c_1\neq 0$. Otherwise $b_0=0$ by Eq.~\eqref{goesqp:eq:101}, it follows that
      $\ket{1,1}\in\range{\rho}$.

      In order to find $\ket{x}\in\myH$  such that $\ket{x,x}\in\range{\rho}$, we consider the following equation:
      \begin{equation}
        \label{goesqp:eq:103}
        d_0
        \begin{pmatrix}
          a_0\\1\\0\\0
        \end{pmatrix}
        +d_1
        \begin{pmatrix}
          0\\0\\b_0\\b_1
        \end{pmatrix}
        + d_2
        \begin{pmatrix}
          c_0\\c_1\\c_0\\c_1
        \end{pmatrix}
        = \ket{z,z}
        =
        \begin{pmatrix}
          z_0^2\\z_0z_1\\z_0z_1\\z_1^2
        \end{pmatrix},
      \end{equation}
      where
      \begin{equation*}
        z =
        \begin{pmatrix}
          z_0\\z_1
        \end{pmatrix}.
      \end{equation*}
      There is a solution for Eq.\eqref{goesqp:eq:103} when
      \begin{equation}
        \label{goesqp:eq:104}
        \begin{split}
          d_0 & = 1,\\
          d_1 & = \frac{c_1 - d_2}{\sqrt{1+c_1^2}},\\
          z_0 & = \frac{\sqrt{(1+c_1^2)(1+c_1d_2)}}{c_1},\\
          z_1 & = c_1 \sqrt{\frac{1+c_1d_2}{1+c_1^2}}.
        \end{split}
      \end{equation}
      It follows that the range of $\rho$ contains a symmetric product vector $\ket{z,z}$.

      It remains to prove the case $\rank{\rho}=4$, which is obvious since its range spans the whole space.
      
      Above all, our proof completes.
    \end{proof}
    Note that for a $2\otimes 2$ completely symmetric state, we can always find a completely symmetric product vector in its range,
    which implies the following corollary.
    \begin{corollary}
      \label{cor:22}
      Any completely symmetric states in $2\otimes 2$ are S-separable.
    \end{corollary}
Furthermore, we  consider the case where $\rank{\rho}=3$.
\begin{lem}
  \label{rank3}
  Any rank 3 completely symmetric states are S-separable.
\end{lem}
\begin{proof}
  It suffices to consider the supported states. Note that  $\rank{\rho_1}\leqslant \rank{\rho}$, the supported
  space for rank 3 symmetric states must be either the
  $3\otimes 3$ or $2\otimes 2$ space.

  For the former case,
  $\rho$ is S-separable by Lemma~\ref{rankN}.

  For the latter one,
   $\rho$ is a $2\otimes 2$ completely symmetric state with rank $3$, which is S-separable by Corollary~\ref{cor:22}.
 \end{proof}

    To sum up, we have the following theorem.
    \begin{thm}
      \label{thm1}
      Let $\rho$ be a completely symmetric state supported in the $N\otimes N$ space, then it is S-separable if either
      \begin{enumerate}
        \item $N\leqslant 2$;
        \item $\rank{\rho} \leqslant N$.
      \end{enumerate}
    \end{thm}
    In addition, we suggest the following  conjecture.
    \begin{conjecture}
      \label{conj2}
     Any separable completely symmetric states are S-separable.
   \end{conjecture}
   If the conjecture is true, then the separability of completely symmetric states will be equivalent to the
   S-separability, which will simplify the study the entanglement on these states on either theoretical or numerical
   aspects.
   
   It should be  expected we can obtain more results compared with Theorem~\ref{thm1} with the separability condition.
   First, we  introduce a useful  lemma.
   
   \begin{lem}
      \label{reduci}
     Suppose $\rho$ is a completely symmetric state. If $\rho$ is  reducible, then $\rho$ is a sum of two completely symmetric states. Moreover, $\rho$ is S-separable if and only if these two states are  both S-separable.
   \end{lem}
  
   Before proving this lemma, we first recall the definition of reducible states in Ref.~\cite{Chen_2011red}.
   \begin{defi}
    Suppose $\rho= \sigma+\delta$, where $\rho,\sigma$, and $\delta$ are the quantum states in the bipartite
    system. Denote by
    $\rho_{1}$  the reduced state $(I\otimes \trace)(\rho)$, similarly for $\sigma_1$ and $\delta_1$. Then
    $\rho$ is said to be reducible 
     if $\range{\rho_{1}}$ is a direct sum of $\range{\sigma_{1}}$ and $\range{\delta_{1}}$. Otherwise, $\rho$ is
     said to be irreducible.
   \end{defi}
   Now we are ready to prove Lemma~\ref{reduci} with the definition described above.
   \begin{proof}[Proof for Lemma~\ref{reduci}]
     Suppose $\rho$ is a reducible state in the $N\otimes N$ system.
     Let $\rho$ be the sum of two quantum states $\rho = \sigma+\delta$, where $\range{\sigma_1}$ and
     $\range{\delta_1}$   are linearly
     independent. Forward, there exists an invertible operator $A$, such that $\range{A\sigma_{1}}$ and
     $\range{A\delta_1}$ are orthogonal. Hence, we can assume that $\range{\sigma_1}$ and
     $\range{\delta_1}$ are orthogonal  without loss of generality. Suppose
     $e_{i},i=0,1,2,\ldots,N-1$ is an orthonormal basis which spans the space $\real^{N}$ and
     $\range{\sigma_{1}}=\text{span}\{e_{0},e_1,\dots,e_{r-1}\},\,r<N$. 

    The first part of our proof completes if we can prove that   $\sigma$ and $\delta$  are both completely symmetric. Furthermore, it suffices to 
     prove $\sigma$ is completely symmetric. In fact, this holds when $\sigma_{2} = (\trace\otimes \I)\sigma$ is also supported on
     the subspace $\text{span}\{e_{0},e_1,\dots,e_{r-1}\}$.

     Since  $\range{\sigma_1}$ and
     $\range{\delta_1}$ are orthogonal, then according to the block representation Eq.~\eqref{goesqp:eq:90},
     \begin{equation}
       \label{goesqp:eq:25}
       \rho =
       \begin{pmatrix}
         \sigma &0\\
         0&\delta
       \end{pmatrix},
     \end{equation}
     which implies that
     \begin{equation}
       \label{eq47}
       \rho_{ijkl}=0, \forall i<r,j\geqslant r.
       \end{equation}

     Moreover $\rho_{ijkl} = \sigma_{ijkl},\forall i,j<r$. By the super symmetry of $\rho$, for $i,j<r$ and $k\geqslant r$ we
     have
     \begin{equation}
       \label{goesqp:eq:26}
       \sigma_{ijkl} = \rho_{ijkl} = \rho_{ikjl} =0.
     \end{equation}
     Similarly, $\sigma_{ijkl}=0$ for $i,j<k$  and $ l\geqslant r$. This follows that
     $\sigma_{2}$  is supported on   $\text{span}\{e_{0},e_1,\dots,e_{r-1},r<N\}$, which implies that $\sigma$ is completely symmetric.

     Consider the second part of Lemma~\ref{reduci}.
     It is obvious that $\rho$ is S-separable when $\sigma$ and $\delta$ are S-separable. Next, we prove  the other
     direction.

     Suppose $\rho = \sigma +\delta$ is a reducible and S-separable state.
    After  applying a suitable   invertible local operator $(A\otimes A)$, $\rho$ can be transformed into the form of
     Eq.~\eqref{goesqp:eq:25}. Suppose 
     $\rho=\sum_{i}\ket{x_i,x_i}\bra{x_i,x_i}$. Let
     \begin{equation*}
       x_i = 
       \begin{pmatrix}
         x_{i0}&x_{i1}&\cdots&x_{i,N-1}
       \end{pmatrix}^{\intercal}.
     \end{equation*}
     By calculating the coefficient $\rho_{klkl}$ and combining Eq.~\eqref{eq47}, we hvae
     \begin{equation*}
       \rho_{klkl} =\sum_i x_{ik}^2x_{il}^2 =0,\;\forall k<r,\,l\geqslant r.
     \end{equation*}
      Hence, for any given $i$, we have  $x_{ik}=0,\forall k<r$ or $x_{il}=0,\forall
     l\geqslant r$. It follows that 
     \begin{equation}
       \label{goesqp:eq:79}
       \begin{split}
         \sigma &= \sum_{x_{ik}=0,\forall k<r}\ket{x_i,x_i}\bra{x_i,x_i},\\
         \delta &= \sum_{x_{il}=0,\forall l\geqslant r}\ket{x_i,x_i}\bra{x_i,x_i}.
         \end{split}
       \end{equation}
       Therefore, $\sigma$ and $\delta$ are both S-separable.
     \end{proof}
     With this lemma, for reducible states, the problem could become simpler by considering each irreducible items respectively.
   \begin{thm}
     \label{rank4}
     Conjecture~\ref{conj2} holds for the completely symmetric states whose ranks do not exceed $4$.
   \end{thm}
   \begin{proof}
     Suppose $\rho$ is a completely symmetric state. If $\rank{\rho}\leqslant 3$, then $\rho$ is S-separable by
     Theorem~\ref{thm1}.
     Hence, we assume that that $\rank{\rho}=4$. Furthermore, if $\rho$ is supported in the $2\otimes 2$ or $4\otimes 4$ space, then 
     $\rho$ is S-separable by Corollary~\ref{cor:22} and Lemma~\ref{rankN}. We  assume further that $\rho$ is supported on the
     $3\otimes 3$ space.

     By Lemma~\ref{real}, $\rho$ is separable over real. We claim that  $\rho$ can be written as a sum of 4 real pure product
     states:
     \begin{equation}
       \label{goesqp:eq:62}
       \rho = \sum_{i=0}^{3}\ket{x_i,y_{i}}\bra{x_i,y_{i}}.
     \end{equation}

     Since $\rho$ is separable over real, there exists a real product vector $\ket{x_1,y_1}\in\bH$ such that $
       \rho - \ket{x_1,y_1}\bra{x_1,y_1}$ is positive, G-invariant,  and has rank $3$, which remains separable over
       $\real$. By reducing  the rank, Eq.~\eqref{goesqp:eq:62} holds.

     If $\ket{x_i},i=0,1,2,3$ are not affinely independent, then $\rho$ is reducible. By Lemma~\ref{reduci} and Theorem~\ref{thm1}, it is
     S-separable. Here we assume that $\ket{x_{i}},i=0,1,2,3$ are affinely independent. Then there exists an invertible
    local  operator such that
     \begin{equation}
       \label{goesqp:eq:110}
       (A\otimes A) \rho (A\otimes A)^{\intercal} = \sum_{i=0}^{2}\ket{i}\bra{i}\otimes \ket{\tilde y_i}\bra{\tilde
         y_{i}}+ \ket{\tilde x_{3}}\bra{\tilde x_{3}}\otimes \ket{\tilde y_{3}}\bra{\tilde y_{3}},
     \end{equation}
     where
     \begin{equation}
       \label{goesqp:eq:119}
      \tilde x_3 = (1,1,1)^{\intercal}.
     \end{equation}
     Let
     \begin{equation}
       \label{goesqp:eq:126}
       \begin{split}
        \tilde   y_0 & =
       \begin{pmatrix}
         a_0\\a_1\\a_2\\
       \end{pmatrix},
     \tilde  y_{1}  =
       \begin{pmatrix}
         b_0\\b_1\\b_2
       \end{pmatrix},
      \tilde y_2  =
       \begin{pmatrix}
         c_0\\c_1\\c_2
       \end{pmatrix},
      \tilde y_3  =
       \begin{pmatrix}
         d_0\\d_1\\d_2
       \end{pmatrix}.
       \end{split}
     \end{equation}
   Let $\sigma= (A\otimes A)\rho(A\otimes A)^{\intercal} $.  By the super symmetry of $\sigma$, we have
     \begin{equation}
       \label{goesqp:eq:131}
       \begin{split}
       \sigma_{0102}& = \sigma_{0201}=\sigma_{1200},\\
       \sigma_{0211}& = \sigma_{1201},\\
       \sigma_{0122}& = \sigma_{1202}.
       \end{split}
     \end{equation}
     Specifically,
     \begin{equation}
       \label{goesqp:eq:132}
       \begin{split}
         d_0d_2& = d_0d_{1} = d_0^2,\\
         d_1^2& = d_0d_1,\\
         d_3^2& = d_0d_2.
       \end{split}
     \end{equation}
     If $d_0=0$, by Eq.~\eqref{goesqp:eq:132}, we have $d_0=d_1=d_2$, which contradicts with our assumption that $\rho$ has
     rank 4. If $d_0\neq 0$, then by Eq.~\eqref{goesqp:eq:132}, we have $d_0=d_1=d_2$. Hence, $\ket{x_3,y_3}\bra{x_3,y_3}$ is
     a completely symmetric state. Forward, $\rho - \ket{x_3,y_3}\bra{x_3,y_3}$ is a completely symmetric state of rank 3, which is
     S-separable by Lemma~\ref{rank3}. Therefor, $\rho$ is S-separable.
   \end{proof}
   Then a natural question arises whether $\rho$ is separable for the rank 4 completely symmetric state. By the \textrm{Chen}'s
   result in Ref.\cite{Chen_2011red}, $\rho$ is separable if and only if $\rho$ contains a product vector in
   its range. We can answer this question if we can check whether the range of a completely symmetric state contains a
   completely symmetric product vector.

   Suggested by Lemma~\ref{rank4}, we now investigate the S-separability of rank $ N+1$ completely symmetric states.
   \begin{thm}
     Suppose $\rho$ is a rank $N+1$ completely symmetric state supported in the $N \otimes N (N>3)$ space. If $\rho$ is
     irreducible,  then $\rho$ is S-separable. Otherwise $\rho = \sigma+\ket{x,x}\bra{x,x}$, where $\sigma$ is a rank
     $N$ state supported
     in the $(N-1)\otimes (N-1)$ space.
   \end{thm}
   \begin{proof}
     First, we discuss the case where $\rho$ is irreducible.
     By Proposition 39 in Ref.~\cite{Chen_2013properties}, $\rho$ is a sum of $N+1$ pure product states. By Lemma~\ref{real},
     there exists a product vector  $\ket{x_1,y_1}\in\real^{N}\otimes\real^{N}$ such that $\sigma=\rho - \ket{x_1,y_1}\bra{x_1,y_1}$ is positive
     and has rank $N$. If $\sigma$ is supported in the $N\otimes N$ space, then by Lemma~\ref{nsum}, $\sigma$ is a sum of $N$
     real pure product states. Otherwise, $\sigma$ must be supported in the $(N-1)\otimes(N-1)$ space. Applying Lemma~\ref{real}
     and Lemma~\ref{nsum} recursively, we conclude that $\rho$ can be written as a sum of $N+1$ real pure product states.
     Suppose
     \begin{equation}
       \label{goesqp:eq:133}
       \rho = \sum_{i=0}^{N}\ket{x_{i}}\bra{x_{i}}\otimes \ket{y_{i}}\bra{y_i},\ket{x_i},\ket{y_i}\in\real^N.
     \end{equation}
     Since $\rho$ is irreducible, $\{\ket{x_i},i=0,1,\ldots,N\}$ is affinely independent. Forward, there exists an
     invertible operator $A$ such that $A\ket{x_i} = \ket{e_i},i=0,1,\ldots,N$, where
     $\ket{e_i}=\ket{i},i=0,1,\ldots,N-1$ and $\ket{e_{N}}$ is the vector whose entries are all $1$'s. Hence, we assume
     that $\ket{x_i}=\ket{e_{i}}$ and let
     \begin{equation}
       \label{goesqp:eq:134}
       \ket{y_i} =
       \begin{pmatrix}
         y_{i0}&y_{i1}&\cdots&y_{i,N-1}
       \end{pmatrix}^{\intercal},i=0,1,,\ldots,N.
     \end{equation}
     By the super symmetry of $\rho$, we have
     \begin{equation}
       \label{goesqp:eq:135}
       \rho_{0i0j}  = \rho_{ij00},i\neq j,i,j> 0.
     \end{equation}
     i.e.,
     \begin{equation}
       \label{goesqp:eq:136}
       y_{N0} y_{Nj}= y_{N0}^2,\forall j> 0.
     \end{equation}
     If $y_{N0}= 0$, note that $\rho_{0jii}=\rho_{ij0i}$, we have $y_{Ni}^{2}=y_{N0}y_{Ni}$ for all $i\neq j,i,j>
     0$. Therefore, $y_{Ni}=0$ for all $i$, which, however,  contradicts with the assumption that $\rho$ has rank $N+1$. If
     $y_{N0}\neq 0$, by Eq.~\eqref{goesqp:eq:136}, $y_{Ni}=y_{N0}$ for all i>0. Hence, $\ket{e_N,y_N}\bra{e_N,y_{N}}$ is completely symmetric. And $\sum_{i=0}^{N}\ket{e_i,y_i}\bra{e_i,y_i}$ is  reducible, by Lemma~\ref{reduci},
     $\ket{e_i,y_i}\bra{e_i,y_i}$ is completely symmetric. Hence, $\rho$ is S-separable and is a sum of $N+1$ supper symmetric
     pure product states.

     If $\rho$ is reducible, then $\rho = \sigma + \ket{z,y}\bra{z,y}$, where $\range{(\I\otimes \trace)\sigma}$ and
     $\{\ket z\}$
     are linearly independent. Hence by Lemma~\ref{reduci}, $ \ket{z,y}\bra{z,y}$ is completely symmetric. Moreover, $\sigma $ is
      supported in the $(N-1)\otimes 
     (N-1)$ space with rank $N$.
   \end{proof}
   Note that in the  above proof, for $\sigma$, the  theorem can be applied recursively when $(N-1)>3$. Therefore, we
   have the following corollary.
   \begin{corollary}
     Let $\rho$ be a rank $N+1$ completely symmetric state which is supported on $N\otimes N (N>3)$ space. Then $\rho$  is either
     a sum of $N+1$ supper symmetric pure product states or $\rho = \sigma +
     \sum_{i=0}^{N-4}\ket{x_i,x_i}\bra{x_i,x_i}$, where $\sigma $  is a rank $4$ completely symmetric state
     supported on  the $3\otimes 3$ space and $\range{\sigma}\bot \ket{x_i},i=0,1,\ldots,N-4$.
   \end{corollary}
   Specifically, $\rho$ in the above corollary will be S-separable if $\rho$ is separable according to Theorem~\ref{thm1}. We
   finish this section by summarizing the results on Conjecture~\ref{conj2}.
   \begin{thm}
     Suppose $\rho$ is a separable completely symmetric state supported in the $N\otimes N$ space. Then Conjecture~\ref{conj2} holds if
     either
     \begin{enumerate}
       \item $N\leqslant 2$;
       \item $\rank{\rho}\leqslant N+1$.
     \end{enumerate}
   \end{thm}

   We finish this section by proposing another interesting question: which kind of quantum states can be transformed into
   the completely symmetric states or even the S-separable states  by the
   local invertible operator. We call these states local-equivalently S-separable. Here we have some examples.
   \begin{example}
     \label{ex0}
     Any rank $N$ states supported in the $N\otimes N$  space can be transformed into S-separable states by the local invertible operator.
   \end{example}
   \begin{proof}
    It was proved in Ref.~\cite{horodecki2000operational} the state can be written as a sum of $N$ pure product
    states. Denoted by
    \begin{equation}
      \label{S-sep-v5.1:eq:21}
      \rho = \sum_{i=0}^{N-1}\ket{x_i,y_i}\bra{x_i,y_i}.
    \end{equation}
    Since $\rho$ is supported in the $N\otimes N$ space, $\ket{x_i}$ and $\ket{y_i}$ are linearly independent,
    respectively. Therefore, we can find invertible operators $A,B\in\complex^{N}$ such that
    \begin{equation}
      \label{S-sep-v5.1:eq:22}
      \begin{split}
        A\ket{x_i}&=\ket{i},i=0,1,\ldots N-1,\\
        B\ket{y_i} &= \ket{i},i=0,1,\ldots N-1.
      \end{split}
    \end{equation}
    Therefore, $(A\otimes B)\rho(A\otimes B)^{\dagger}$ is S-separable.
  \end{proof}
  \begin{example}
    \label{ex01}
    Let $\rho$ be a separable state in the $N\otimes N $ system. And
    \begin{equation}
      \label{S-sep-v5.1:eq:23}
      \rho = \sum_{i=0}^{N}\ket{x_i,y_i}\bra{x_i,y_i},
    \end{equation}
    where $\ket{x_i}$ and $\ket{y_i}$ are in general position, then $\rho$ is local-equivalently S-separable.
  \end{example}
  \begin{proof}
    Since $\ket{x_i}$ are in general position, then $\ket{x_i},i=0,\ldots,N-1$ are a basis for $\complex^{N}$. Moreover,
    there exists a local invertible operator $A\in\complex^{N\times N}$  such that
    \begin{equation}
      (A\otimes \I)\rho (A\otimes \I)^{\dagger} = \sum_{i=0}^{N-1}\ket{i,y_i}\bra{i,y_{i}}+\ket{x_N,y_N}\bra{x_N,y_{N}},
    \end{equation}
    where $x_{N}$ is the vector whose entries are all $1$'s. Note that $y_i,i=0,\ldots,N-1$ are linearly independent,
    there exists an invertible operator $B$ such that $B\ket{y_i}=\ket{i},i=0,1,\ldots,N-1$ and $B\ket{y_{N}}$ is the vector
    whose entries are all non-zero, denoted by
    \begin{equation}
      \label{S-sep-v5.1:eq:25}
      \ket{y_N} =
      \begin{pmatrix}
        a_0&a_1&\cdots&a_{N-1}
      \end{pmatrix}^{\intercal}.
    \end{equation}
    Let
    \begin{equation}
      \label{S-sep-v5.1:eq:26}
      C =
      \begin{pmatrix}
        \frac{1}{a_0}&0&\cdots&0\\
        0&\frac{1}{a_2}&\cdots&0\\
        \vdots&\vdots&\ddots&\vdots\\
        0&0&\cdots &\frac{1}{a_{N-1}}
      \end{pmatrix}.
    \end{equation}
    Then
    \begin{equation}
      \label{S-sep-v5.1:eq:27}
      \begin{split}      C\ket{y_N} &= \ket{x_N},\\
        C\ket{i} & = \frac{1}{a_i}\ket{i},i=0,1,\ldots,N-1.
 \end{split}
    \end{equation}
    Hence,
    \begin{equation}
      \label{S-sep-v5.1:eq:28}
      \begin{split}
      &(A\otimes CB)\rho(A\otimes CB)^{\dagger} \\
      & = \sum_{i=0}^{N-1}\frac{1}{\abs{a_i}^2}\ket{i,i}\bra{i,i} + \ket{x_N,x_N}\bra{x_N,x_{N}}.
      \end{split}
    \end{equation}
    Therefore, $\rho$ is local-equivalently S-separable.
  \end{proof}
    \section{Best completely symmetric approximation}%
    \label{sec:real-super-symmetric}

    In this section, we consider the problem whether a given completely symmetric state is S-separable. We propose a
    numerical method to solve this problem.

    Let $\symmf$ be the set of completely symmetric states with trace one.  Let
    \begin{equation}
      \label{goesqp:eq:107}
      \symmg = \mathrm{conv}\{xx^{\intercal}\otimes xx^{\intercal},x\in\myH,\norm{x}=1\}.
    \end{equation}
    If $\rho$ is S-separable, then it must be contained in the set $\symmg$.

    Given any $\rho\in\symmf$, we consider the following optimization problem:
    \begin{equation}
      \label{cap}
      \begin{split}
        \min_{\sigma}        & \quad  \frac{1}{2}\norm{\rho - \sigma}^{2}\\
        \text{ s.t.} & \quad {\sigma\in \symmg}.
      \end{split}
    \end{equation}

    Note that $\rho\in\symmg$ if and only if the optimal value of the above optimization problem equals zero. Otherwise,
    we can find its closest S-separable state.

    Since $\symmg$ is a compact convex set and the objective function of \eqref{cap} is a strictly convex quandratic
    function of $\sigma$, the optimization problem has a unique minimizer $\rho_{*}$, which is essentially the
    projection of $\rho$ to $\symmg$.
    
    \newcommand{\proj}{ \mathrm{Proj}_{\symmg}} Denoted by
    \begin{equation}
      \label{goesqp:eq:77}
      \rho_{*} = \mathrm{Proj}_{\symmg}(\rho).
    \end{equation}
    We have the following properties~\cite{bertsekas2005dynamic}:
    \begin{lem}
      Suppose $\rho$ is completely symmetric and $\rho_{*}\in\symmg$. Then the following statements are equivalent:
      \begin{enumerate}
        \item $\rho_{*} = \proj(\rho)$;
        \item $\langle\rho-\rho_{*},\sigma-\rho_{*}\rangle\leqslant 0,\;\forall \sigma\in \symmg$;
        \item The following inequality holds:
              \begin{equation}
                \label{ineq}
                \langle\rho-\rho_{*},xx^{\intercal}\otimes
                xx^{\intercal}-\rho_{*}\rangle\leqslant 0,
              \end{equation}
              for all $x\in\real^{N},\norm{x}=1$.
      \end{enumerate}
    \end{lem}
    \begin{proof}
        First we prove  (1)$\Rightarrow$ (2).

              Let $\delta(\lambda) = (1-\lambda)\rho_{*}+\lambda\sigma$ for $\sigma\in\symmg$ and
              $\lambda\in[0,1]$. Then $\delta(\lambda)\in\symmg$.

             Define $m(\lambda) = \norm{\rho -\delta(\lambda) }^2_{F}$, then
              \begin{equation}
                \label{S-sep-v5.1:eq:16}
                \begin{split}
                  g(\lambda) & = \norm{(1-\lambda)(\rho-\rho_{*})+\lambda(\rho-\sigma)}^2_F\\
                  & =(1-\lambda)^2 \norm{\rho-\rho_{*}}_F^2+\lambda^2\norm{\rho-\sigma}^2_F\\
                  &\quad +  2\lambda(1-\lambda)\langle \rho-\rho_{*},\rho-\sigma\rangle,
                \end{split}
              \end{equation}
              which is quadratic and
              \begin{equation}
                \label{S-sep-v5.1:eq:17}
                m'_{+}(0) = -2\langle \rho-\rho_{*},\sigma-\rho_{*}\rangle.
              \end{equation}
              Since $\rho_{*}$ is the global minimizer of $\norm{\rho-\sigma}$, $m'_{+}(0)$ must be positive, i.e,
              \begin{equation}
                \label{S-sep-v5.1:eq:18}
                \langle\rho-\rho_{*},\sigma-\rho_{*}\rangle\leqslant 0,\;\forall \sigma\in \symmg.
              \end{equation}

              Next, we prove (2)$\Rightarrow$(1).

              Assume (2) holds, then
              \begin{equation}
                \label{S-sep-v5.1:eq:19}
                \begin{split}
                  \norm{\rho-\sigma}^2_F & = \norm{\rho-\rho_{*}+(\rho_{*}-\sigma)}_F^2\\
                  & = \norm{\rho-\rho_{*}}_F^2+\norm{\rho_{*}-\sigma}_F^2\\
                  &\quad -2\langle \rho-\rho_{*},\sigma-\rho_{*}\rangle\\
                  & \geqslant \norm{\rho-\rho_{*}}_F^2+\norm{\rho_{*}-\sigma}_F^2\\
                  &\geqslant \norm{\rho-\rho_{*}}_F^2,
                \end{split}
              \end{equation}
              which implies that $\rho_*=\proj(\rho)$.
              
              Since any $\sigma\in\symmg$ can be represented as a convex combination of $xx^{\intercal}\otimes
              xx^{\intercal},x\in\real^{N}$, denoted by $\sigma = \sum_i \lambda_i x_ix_i^{\intercal}\otimes
              x_ix_i^{\intercal}$, we have
              \begin{equation}
                \label{S-sep-v5.1:eq:20}
                \langle \rho-\rho_{*},\sigma-\rho_*\rangle = \sum_i \lambda_i \langle \rho-\rho_{*}, x_ix_i^{\intercal}\otimes
              x_ix_i^{\intercal}-\rho_*\rangle.
              \end{equation}
              Therefore, (2) and (3) are equivalent as well.              
    \end{proof}

    Recall the density approximation algorithm in Ref.~\cite{dahl2007tensor}.  The algorithm is an adaption of a general
    algorithm in convex programming called the Frank-Wolfe method. By the similar idea, we propose an algorithm that can
    be used to 
    find the closest S-separable states.

    Let $\rho_0$ be a candidate for approximating $\rho_{*}$. If Eq.~\eqref{ineq} satisfies, then $\rho_0$ is the closest
    S-separable state of $\rho$, i.e., the solution of optimization problem~\eqref{cap}. Otherwise, $\sigma-\rho_0$
    forms a feasible descent direction provided $\sigma$ violates Eq.~\eqref{ineq}. Hence, the key point is to check whether
    $\langle\rho-\rho_{*},\sigma-\rho_{*}\rangle$ is greater than $0$. This can be solved by finding the solution of the
    following optimization problem
    \begin{equation}
      \label{goesqp:eq:108}
      \begin{split}
        \max_{\sigma} & \quad \langle\rho-\rho_0,\sigma-\rho_0\rangle\\
        \text{ s.t.} & \quad\sigma\in\symmg,
      \end{split}
    \end{equation}
    which is equivalent to
    \begin{equation}
      \label{goesqp:eq:98}
      \begin{split}
        \max_{x}& \quad \bra{x,x}\eta \ket{x,x}\\
        \text{ s.t. }&\quad\norm{x}=1,x\in\real^{N},
      \end{split}
    \end{equation}
    where  $\eta = \rho-\rho_0$.
    
     In the next section, we 
    propose an algorithm to solve this sub-optimization problem~\eqref{goesqp:eq:98}.

     Let 
    \begin{equation}
      \label{goesqp:eq:111}
      \begin{split}
        f(x) &=\frac{1}{4} \bra{x,x}\eta \ket{x,x}.\\
      \end{split}
    \end{equation}

   Denote by $x_0$ the global maximizer of~\eqref{goesqp:eq:98}. If
    $4f(x_0)>~\langle\eta,\rho_0\rangle$, then $\sigma(x_0)-\rho_0$
    is a decent direction.  We can replace the candidate $\rho_0$ with $\rho_0 + \alpha (\sigma(x_0)-\rho_0)$. Here
    $\alpha$ is chosen such that the objective function in \eqref{cap} has the minimum value, which is
    \begin{equation}
      \label{goesqp:eq:112}
      \alpha = \frac{\langle \rho-\rho_0,\sigma(x_0)-\rho_0\rangle}{\norm{\sigma(x_0)-\rho_0}_F}.
    \end{equation}
    If $4f(x_0)\leqslant \langle\eta,\rho_0\rangle $, then we can conclude that $\rho_0$ is the solution of the
    optimization problem~\eqref{cap}.

    The algorithm for solving the optimization problem~\eqref{cap} is described as  follows~\cite{bertsekas1999nonlinear}: 
    \begin{algorithm}[H]
      \centering
      \caption{Frank-Wolfe method for solving~\eqref{cap}}
      \label{mainalgorithm}
      \begin{algorithmic}[1]
        \Require $\rho$: completely symmetric state, $\varepsilon_1$: tolerance for terminating algorithm.
        \Ensure $\rho_{*}$:
        Closest S-separable state of $\rho$.
        \State{}Choosing initial candidate $\rho_0$
        \State $err \gets +\infty$
        \State $k=0$
        \While{          $err > \varepsilon_1 $ }
              \State $k\gets k+1$
              \State $\eta\gets \rho-\rho_k$
              \State Obtain $x_{k+1}$ by solving optimization problem~\eqref{goesqp:eq:98}
              \State $err \gets \norm{\rho_k-\rho_{k-1}}$
              \If{$4f(x_{k+1})\leqslant \langle\eta,\rho_k\rangle $}
                     \State \Return $\rho_k$
                     \Else
                     \State $\alpha_{k+1}\gets {\langle \rho-\rho_k,\sigma(x_{k+1})-\rho_k\rangle}/{\norm{\sigma(x_{k+1})-\rho_k}_{F}}$
                     \State
                     $\rho_{k+1}\,\gets \rho_k+ \alpha_{k+1}(\sigma(x_k)-\rho_k)$\newline \Comment{$\mathtt{\sigma(x) =
                       xx^{\intercal}\otimes xx^{\intercal}}$}
                     
              \EndIf
        \EndWhile
        \State $\rho_{*}\gets\rho_k$      
      \end{algorithmic}
    \end{algorithm}

    According to the convergence theorem in Ref.~\cite{bazaraa2013nonlinear}, the sequence of iterative points generated
    by Algorithm~\ref{mainalgorithm} converges to $\proj(\rho)$ globally. We have the 
    following theorem.  
    \begin{thm}
      The sequence $\rho_k$ generated by Algorithm~\ref{mainalgorithm} converges to the optimal solution $\rho_{*}$ of~\eqref{cap}.
    \end{thm}

    Among the above algorithm, the most important subroutine per iteration   is to find a global optimizer of
    the optimization problem~\eqref{goesqp:eq:98}. In the following section, we suggest a numerical method to solve
    this problem.

\section{algorithm for solving optimization problem (\ref{goesqp:eq:98})}%
\label{sec:algor-optim-probl}
Before proposing our algorithm, we  give the KKT conditions for the optimization
problem~\eqref{goesqp:eq:98}.
We begin this section with computing the gradient and Hessian matrix of $f(x)$, where $f(x) =\frac{1}{4}\langle x,x
|\eta| x,x\rangle$   is given in 
Eq.~\eqref{goesqp:eq:111}. 

Let $\nabla f(x)$ denote the gradient vector and $\nabla^2f(x)$ the Hessian matrix at the point $x$.

We have
\begin{equation}
  \label{goesqp:eq:22}
  \begin{split}
    \nabla f (x) & =   B_{x} x,\\
    \nabla^2f(x) &  = 3 B_x,
    \end{split}
  \end{equation}
  where
  \renewcommand{\adj}[1]{{#1}^{\intercal}}
  \begin{equation}
    \label{goesqp:eq:23}
    \begin{split}
    B_x &=  \begin{pmatrix}
      \adj x \eta_{00} x & \cdots & \adj x \eta_{0,N-1} x\\
      \vdots & \ddots & \vdots \\
      \adj x \eta_{N-1, 0} x & \cdots & \adj x \eta_{N-1, N-1} x
    \end{pmatrix},
    \end{split}
  \end{equation}
  according to the block representation of $\eta$ as in  Eq.~\eqref{goesqp:eq:90} and \eqref{goesqp:eq:91}.
  By the symmetry of $\eta$,
  \begin{equation}
    \label{goesqp:eq:122}
    y^{\intercal}\nabla^2f(x) y = x^{\intercal}\nabla^{2}f(y)x.
  \end{equation}

Let
\begin{equation}
  \label{goesqp:eq:114}
  \Lag(x,\lambda) = f(x) - \lambda \frac{x^Tx-1}{2}
\end{equation}
be the Lagrange function. We have
\begin{equation}
  \label{goesqp:eq:115}
  \begin{split}
    \nabla_{x}\Lag(x,\lambda) &= \nabla f(x) - \lambda x =B_xx-\lambda x,\\
    \nabla_{xx}^2\Lag(x,\lambda)& = \nabla^2f(x) - \lambda\I = 3B_x - \lambda\I.
  \end{split}
\end{equation}
 Furthermore, we have the following KKT conditions.
\begin{lem}[First-Order KKT condition~\cite{Pagonis2006}]
  If $x_{*}$ is a local optimizer of Eq.~\ref{goesqp:eq:98}, then there exists a $\lambda^{*}\in\real$ such that the following KKT
  conditions hold
  \begin{equation}
    \label{goesqp:eq:14}
    \begin{split}
      \nabla_f(x_{*}) - \lambda_{*} x_{*}  & = 0,\\
      {x_{*}}^\intercal x_{*}  - 1 &= 0.
      \end{split}
    \end{equation}
     \end{lem}
  Note that Eq.~\eqref{goesqp:eq:14} can be written  explicitly as
    \begin{equation}
      \label{goesqp:eq:24}
      \begin{split}
         B_{x_{*}}x_{*} &= \lambda_{*} x_{*} ,\\
        x_{*}^{\intercal} x_{*}& = 1.
      \end{split}
    \end{equation}
   
 It follows from Eq.~\eqref{goesqp:eq:24} that
  \begin{equation}
    \label{goesqp:eq:28}
    \lambda_{*} =  x^{\intercal}_{*} B_{x_{*}}x_{*} =  4f(x_{*}),
  \end{equation}
which  guarantees the uniqueness of the associated Lagrange multiplier.
  \begin{lem}[Second-Order Necessary Conditions\cite{Pagonis2006}]
    If $x_{*}$ is a local optimizer of Eq.~\eqref{goesqp:eq:98} and $\lambda_{*}$ is the associated Lagrange multiplier, then
    the second order optimality condition holds
    \begin{equation}
      \label{goesqp:eq:27}
      d^{\intercal}\nabla_{xx}^2\Lag(x_{*},\lambda_{*})d\leqslant 0,\, \forall d^{\intercal}x_{*}=0.
    \end{equation}
  \end{lem}
  That is
  \begin{equation}
    \label{goesqp:eq:29}
    d^{\intercal}\nabla^2f(x_{*})d \leqslant \lambda_{*}, \forall d^{\intercal}x_{*} = 0,\norm{d}=1.
  \end{equation}

  Sufficient condition is the condition on $f(x)$ that ensures $x_*$   is a local maximum to the
  problem Eq.~\eqref{goesqp:eq:98}. It differs in that the inequality in Eq.~\eqref{goesqp:eq:29} is replaced by a strict
  inequality. 
  \begin{lem}[Second-Order Sufficient Condition \cite{Pagonis2006}]
    Suppose $(x_{*},\lambda_{*})$  satisfy the first order KKT condition Eq.~\eqref{goesqp:eq:24}. Suppose also that
    \begin{equation}
      \label{goesqp:eq:30}
      d^{\intercal} \nabla^2f(x_{*})d < \lambda_{*}, \forall d^{\intercal}x_{*}=0,\norm{d}=1.
    \end{equation}
    Then $x_{*}$ is a local maximizer of the optimization problem.
  \end{lem}

   Followed by the second-order sufficient  condition, we have the following inequality.
  \begin{lem}
      \label{lem4}
      Suppose $x_*$ is a local maximum of Eq.~\eqref{goesqp:eq:98}, which satisfies the second-order sufficient KKT
      condition. Then there exists a positive 
      constant $a_1$ such that
      \begin{equation}
        \label{goesqp:eq:45}
        \begin{split}
        \left| (x-x_{*})^{\intercal}(\lambda\I-\nabla^2 f(x_{*}))(x-x_{*})\right|
        > a_1\norm{x-x_{*}}^2,
        \end{split}
      \end{equation}
      for all  $x\neq \pm x_{*},\norm{x}=1$.
    \end{lem}
    \begin{proof}
        \newcommand{\myF}{\mathcal{F}}
  Let
  \begin{equation}
    \label{goesqp:eq:116}
    \myF = \{d:\langle d,x_{*}\rangle =0,\norm{d}=1\}.
  \end{equation}
  Then $\myF$ is a  compact set, which implies that $ d^{\intercal} \nabla^2f(x_{*})d$ can obtain
  its maximum over $\myF$, namely $p$. By Eq.~\eqref{goesqp:eq:30}, we have
  \begin{equation}
    \label{goesqp:eq:31}
   \lambda_{*}-  d^{\intercal} \nabla^2f(x_{*})d > \lambda_{*} - p >0 , \forall d\in\myF.
 \end{equation}
 Let $a =  \lambda_{*} - p $, then
 \begin{equation}
   \label{goesqp:eq:41}
      \lambda_{*}-  d^{\intercal} \nabla^2f(x_{*})d > a , \forall d\in\myF.
 \end{equation}
 On the other hand, by Eq.~\eqref{goesqp:eq:28} we have,
 \begin{equation}
   \label{goesqp:eq:33}
   \begin{split}
     \lambda_{*} - x_{*}^{\intercal}\nabla^2f( x_{*}) x_{*}& = \lambda_{*} - 3 \lambda_{*}\\
     &= - 2 \lambda_{*}\\
     &=-8f(x_{*}).
     \end{split}
   \end{equation}
   
Since $x_{*}$ is an eigenvector of $\nabla^2f(x_{*})$, therefore,
 \begin{equation}
   \label{goesqp:eq:32}
   | \lambda_{*}-  d^{\intercal} \nabla^2f(x_{*})d |\geqslant \min\{ a,2f(x_{*}) \} , \forall \norm{d}=1.
 \end{equation}
 Our proof completes by replacing $d$ with $(x-x_{*})/\norm{x-x_{*}}$ for the case $f(x_*)\neq 0$.

 If $f(x_{*})=0$, let
 \begin{equation}
   x = \cos(\theta)x_{*}+ \sin(\theta)x_{\bot}, x_{\bot}\in\myF,\theta\in[0,2\pi].
   \end{equation}By
 calculation, we have
  \begin{equation}
    \label{goesqp:eq:43}
   \abs{ \sin(\theta)} = \sqrt{\norm{x-x_{*}}^{2} - \frac{\norm{x-x_{*}}^4}{4}}.
  \end{equation}
  For simplicity, we place the proof of the above equality to the Appendix~\ref{apdeq1}.
  
  Note that $x_{*}^{\intercal}\nabla^{2} f(x_{*})x_{\bot}=0$ and $\lambda_{*}=0$, hence
  \begin{equation}
    \label{goesqp:eq:44}
    \begin{split}
     | {x-x_{*}}^{\intercal}(\lambda_{*}\I -\nabla^2f(x_{*}))(x-x_{*})| 
      =&
      \sin(\theta)^2|x_{\bot}^{\intercal}\nabla^2f(x_{*}) x_{\bot}|\\
      \geqslant & a \norm{x-x_{*}}^2(1-\frac{\norm{x-x_{*}}^2}{4}).
      \end{split}
    \end{equation}
    And $\norm{x-x_{*}}< {2}$, hence we can choose an $a_1>0$ such that Eq.~\eqref{goesqp:eq:45} satisfied.
       \end{proof}

The above lemma are used in the analysis of the convergence for  the Algorithm~\ref{combineNewton}. Furthermore, we investigate
the properties of the Hessian matrix $\nabla^{2}f(x)$. We  show that this Hessian
matrix of $f(x)$ is bounded.
 \begin{lem}
   Let $\nabla^2f(x)$ be the Hessian matrix of $f$ at point $x$, then for any unit vector $x$, 
   \begin{equation}
     \label{goesqp:eq:38}
     \norm{\nabla^2f(x)}_{2} \leqslant 3\norm{\eta}_{2}.
   \end{equation}
 \end{lem}
 \begin{proof}
   By the definition of $l_2$-norm, we have
   \begin{equation}
     \label{goesqp:eq:39}
     \norm{\nabla^2f(x)}_{2} = \max_{\norm{y}=1} y^{\intercal}(\nabla^2f(x))y.  
   \end{equation}
   And 
   \begin{equation}
     \label{goesqp:eq:40}
     \begin{split}
       y^{\intercal}\nabla^2f(x)y &= 3y^{\intercal}B_{x}y\\
       &= 3(x\otimes y)^{\intercal} A (x\otimes y)\\
       &\leqslant 3\norm{\eta}_{2}\norm{x}^{2}\norm{y}^{2}.
     \end{split}
   \end{equation}
   Therefore, Eq.~\eqref{goesqp:eq:38} holds.
 \end{proof}
  Moreover, the Hessian
 matrix $\nabla^2f(x)$ is proved to be Lipschitz in the following lemma.
  \begin{lem}
    The Hessian matrix $\nabla^{2}f(x)$  is Lipschitz continuous.
  \end{lem}
  \begin{proof}
    Let $x,y$ be any two unit vectors, then
    \begin{equation}
      \label{goesqp:eq:34}
      \begin{split}
  \norm{\nabla^2f(x) - \nabla^2f(y)}_2&
         = \max_{\norm{z}=1} (z^{\intercal}\nabla^2f(x)z -z^{\intercal}\nabla^2f(y) z).
        \end{split}
      \end{equation}
      By the symmetry of $\eta$ as in Eq.~\eqref{goesqp:eq:122},
      \begin{equation}
        \label{goesqp:eq:120}
        \begin{split}
          \norm{\nabla^2f(x) - \nabla^2f(y)}_{2}
          & = \max_{\norm{z} =
          1}(x^{\intercal}\nabla^2f(z)x-y^{\intercal}\nabla^2f(z)y)\\
        &= (x-y)^{\intercal}\nabla^{2}f(z)(x+y)\\
        &\leqslant \norm{\nabla^{2} f(z)}_{2}\norm{x-y}\norm{x+y}\\
        &\leqslant 2 \norm{\nabla^{2} f(z)}_{2}\norm{x-y}.
        \end{split}
      \end{equation}
      Therefore, by Eq.~\eqref{goesqp:eq:38},
      \begin{equation}
        \label{goesqp:eq:124}
              \norm{\nabla^2f(x) - \nabla^2f(y)}_{2}   \leqslant 6\norm{\eta}_{2}\norm{x-y}.
      \end{equation}
  \end{proof}
  \begin{lem}
    \label{lem7}
    Suppose $x_{*}$ is a local maximizer of Eq.~\eqref{goesqp:eq:98}. Then for any unit vector $x$, there exists a positive constant $a_2$
    such that 
   $ |f(x_*) - f(x)| \leqslant a_2 \norm{x-x_{*}}^{2}$.
  \end{lem}
  \begin{proof}
    Let $\lambda_{*}$ be the associated Lagrange multiplier which satisfies Eq.~\eqref{goesqp:eq:24}.
    By the Taylor expansion of $\Lag(x,\lambda) $ at point $ (x_*,\lambda_*)$, we have
    \begin{equation}
      \label{goesqp:eq:35}
      f(x) = f(x_{*}) +\frac{1}{2} (x-x_{*})^{\intercal}(\nabla^2f(\xi)-\lambda_{*}\I)(x-x_{*}),
    \end{equation}
    where $\xi =tx+(1-t)x_{*}$, $0<t<1$.
    Hence,
    \begin{equation}
      \label{goesqp:eq:36}
      \begin{split}
       \abs{ f(x_{*})- f(x)}
        & \leqslant\frac{1}{2} \norm{\nabla^2f(\xi)-\lambda_{*}\I}_{2}\norm{x-x_{*}}^{2}\\
        & \leqslant \frac{1}{2}(\norm{\nabla^2f(\xi)}_{2}+|\lambda_{*}|) \norm{x-x_{*}}^{2}.
      \end{split}
    \end{equation}
    Since $|\lambda_{*}| \leqslant \norm{\eta}_2$, we have 
    \begin{equation}
      \label{goesqp:eq:37}
      \abs{f(x_{*})-f(x)}\leqslant a_2\norm{x-x_*}^{2},
    \end{equation}
    where $a_2\leqslant \frac{1}{2}(3\norm{\eta}_2+\lambda_{*})\leqslant 2\norm{\eta}_2$.
  \end{proof}

  In the following subsections, we suggest the method for solving the optimization problem Eq.~\eqref{goesqp:eq:98}. We begin
  with the simplest case where $N=2$, where the exact solution can be deduced by finding the roots of a fourth order
  polynomial~\cite{qi2009z}.

\subsection{Direct method for $N=2$}%
\label{sec:modif-grad-meth}
Suppose we have two vector $x,y$, let
\begin{equation}
  v(t_1,t_2) = \frac{t_{1}x+t_{2}y}{\norm{t_{1}x+t_{2}y}},t_1,t_2\in\real,
\end{equation}
  We want to find
the maximal value of $ f(v({t_1},t_{2}))$.

For simplicity, assume $x,y$ are unit vectors.

If $t_1\neq 0$, we can assume that $t_1=1$, and $t_2=t$, let
\begin{equation}
  \label{symmetrictensor:eq:3}
  \begin{split}
    g(t) &=f(v(1,t)),\\
    &=\frac{1}{4(1+2bt+t^2)^2} \trans{(x+ty)} B_{(x+ty)}\cdot (x+ty),\\
    & = \frac{1}{(1+2bt+t^2)^2}\left( p_4t^4+ p_3t^3+p_2t^2+p_1t+p_0 \right),\\
    & = \frac{P(t)}{(1+2bt+t^2)^{2}},
  \end{split}
\end{equation}
where
\begin{align*}
  p_0 & = f(x),\\
  p_1 & = \trans x B_{x}\,y,\\
  p_2 & = \frac{3}{2}\trans x B_{y}\, x,\\
  p_3 & =  \trans x B_{y}\, y,\\
  p_4 & =f(y),\\
  P(t) & = p_4t^4+ p_3t^3+p_2t^2+p_1t+p_0 .
\end{align*}
Note that $g(0) = f(x)$  and $g(\infty) = f(y)$. The maximal value exists where $g'(t)=0$:
\begin{equation}
  \label{symmetrictensor:eq:4}
  \dfrac{P'(t)(1+2bt + t^2) - 4  P(t) (b+t)}{(1+2bt+t^{2})^{3}} = 0.
\end{equation}
which is equivalent to
\begin{equation}
  \label{symmetrictensor:eq:5}
  \begin{split}
    p_1& - 4 p_0 b + ( 2 p_1 - 2 p_1 b-4 p_0 ) t \\
&    + ( 3 p_3 + 4 p_1 b -3 p_1  - 
    4 p_2 b) t^2\\
   & + (2 p_1 - 4 p_2 + 2 p_3 b+ 4 p_4 ) t^3 + ( 
   4 p_4 b-p_3 ) t^4\\
   &=0\\
    \text{ or \qquad} t &=\infty.
    \end{split}
\end{equation}
Compare the values of  $g(t)$ at these real solutions $t_i$ of Eq.~\eqref{symmetrictensor:eq:5}. Suppose $g(t_*)$ has the
maximal value and $v_{*}= v(1,t_{*})$,  then $\frac{v_{*}}{\norm{v_{*}}}$ is the maximal solution over the  subspace $\{ x ,y \}$.

\subsection{Power Method}%
\label{sec:algorithm}
In this subsection, we will introduce the power method utilized for solving the optimization problem~\eqref{goesqp:eq:98}.
In order to make the algorithm converge, we make a  translation to the objective function.

Let \begin{equation}
  h(x) = f(x) + \frac{\alpha}{2}((\adj x x)-1),
  \end{equation} where $\alpha_0$ is an undetermined positive constant.

Then
\begin{equation}
  \label{goesqp:eq:125}
  \begin{split}
    \nabla h(x) &= \nabla f(x) + \alpha x,\\
    \nabla^2h(x)& = \nabla^{2} f(x) +\alpha \I.
  \end{split}
\end{equation}

Since $\norm{\nabla^{2}h(x)}_{2}$ is bounded,  $\nabla^{2}h(x)$ is positive for any $x$ when $\alpha$ is large enough. If $\nabla^{2}h(x)$  is
positive, then the iterations generated by the following algorithm will be monotonically increasing
\begin{algorithm}[H]
      \caption{Power Method for  optimization problem~\eqref{goesqp:eq:98}}
      \label{PM}
      \begin{algorithmic}[1]
        \Require $\eta$: completely symmetric matrix, $\rho-\rho_{k}$ in $k$-th step in Algorithm~\ref{mainalgorithm}; $\varepsilon_{2}$: tolerance for terminate algorithm
        \Ensure $(x_{*},f(x_{*}))$: maximizer and maximal value  of $f$ over unit sphere
        \State Choosing initial candidate $x_0$
        \State Choosing $\alpha$ such that $3B_{x}+\alpha\I>0 $ for all $x$ \Comment{One possible choice is 
          $3\norm\eta_{2}$, See Eq.~\eqref{goesqp:eq:23} for $B_{x}$}
        \State $err \gets +\infty$     
        \State $k=0$
        \While{$err > \varepsilon_{2} $ }
               \State $k\quad\;\;\gets k+1$               
              \State $d\quad\;\;\gets B_{x_{k}}x_{k}+\alpha x_{k} $ 
              \State $x_{k+1}\gets d/\norm{d}_{2}$
          
              \State $err\;\;\, \gets \norm{x_k-x_{k+1}}$ 
       \EndWhile
       \State $x_{*}\gets x_{k}$
       \State $f(x_{*})\gets \bra{x_{k},x_{k}}\eta\ket{x_{k},x_{k}}/4$
      \end{algorithmic}
    \end{algorithm}

    \subsection{SQP algorithm}
    In this subsection, we introduce a locally quadratically convergent  algorithm for solving the  optimization
    problem~\eqref{goesqp:eq:98}.
    
Consider the  Lagrangian function 
\begin{equation}
  \label{symmetrictensor:eq:26}
  \Lag(x,\lambda) = f(x) - \lambda \frac{x^{\intercal}x-1}{2}.
\end{equation}
At a current iteration $x_{k}$, a basic sequential quadratic programming algorithm defines an appropriate
search direction $p_{k}$  as a solution to the quadratic programming subproblem:
\begin{equation}
  \label{goesqp:eq:130}
  \begin{split}
    \max_{p}&\quad f(x_k)+\nabla f(x_k)^{\intercal}p + \frac{1}{2}p^{\intercal}\nabla^{2}_{xx}\Lag(x_{k},\lambda_{k})p\\
    \text{s.t.} &\quad x_{k}^{\intercal}p=0.
  \end{split}
\end{equation}

The first order KKT condition of the equality-constrained problem Eq.~\eqref{goesqp:eq:98} can be written as a system of $N+1$
equations:
\begin{equation}
  \label{symmetrictensor:eq:27}
  F(x,\lambda) =
  \begin{pmatrix}
    \nabla f(x) - \lambda x \\
    x^{\intercal}x-1 
  \end{pmatrix}=0.
\end{equation}
The SQP step~\eqref{goesqp:eq:130} is equivalent to Newton's method applied to the above nonlinear system.

If $x_k$ is the current iteration, which we assume that $x^{\intercal }x  =1$, and $\lambda_{k}$ is
the current approximation to the multiplier associated with the $x_k$.

The Newton step from the recent iteration  is  expressed by 
\begin{equation}
  \label{symmetrictensor:eq:29}
  \begin{split}
  \begin{pmatrix}
    x_\nt\\
    \lambda_{\nt} 
  \end{pmatrix}
  & =
  \begin{pmatrix}
    x_k\\
    \lambda_k 
  \end{pmatrix}
  +
  \begin{pmatrix}
    p_k\\
    \Delta{\lambda}_{k}
  \end{pmatrix},
  \end{split}
\end{equation}

where $p_k$ and $p_{\lambda}$ are obtained by solving the linear system
\begin{equation}
  \label{symmetrictensor:eq:30}
  F'(x_k,\lambda_k)
  \begin{pmatrix}
    p_k\\
    \Delta\lambda_{k}
  \end{pmatrix}
  =-
  \begin{pmatrix}
     \nabla f(x_{k}) - \lambda_k x_k\\
    0
  \end{pmatrix}.
\end{equation}
Here $F'(x_k,\lambda_{k})$ is the Hessian matrix of $\Lag(x,\lambda)$ at point $(x_k,\lambda_k)$:
\begin{equation}
  \label{symmetrictensor:eq:28}
  \begin{split}
  F'(x_k,\lambda_k) &=
  \begin{pmatrix}
    \nabla^2_{xx} \Lag(x_{k},\lambda_{k}) & -x_{k}\\
   - x_{k}^{\intercal} &0
  \end{pmatrix}\\
  &=
  \begin{pmatrix}
    \nabla^2f(x_{k})-\lambda_k\I&-x_{k}\\
    -x_k^{\intercal}&0
  \end{pmatrix}.
  \end{split}
\end{equation}
Let
\begin{equation}
  \label{goesqp:eq:117}
  x_{\sqp} = \frac{x_\nt}{\norm{x_\nt}}.
\end{equation}
 To obtain an estimate for the multiplier, instead of using Eq.~\eqref{symmetrictensor:eq:29} directly,  we minimize the norm of the residual
of $\nabla f(x_{\sqp}) - \lambda x_{\sqp}$, this gives
\begin{equation}
  \label{symmetrictensor:eq:32}
  \lambda_{k+1} = \nabla f(x_{\sqp})^{\intercal} x_{\sqp}.
\end{equation}

However, this SQP method may not converge when $x_0$ is far away from the local maximizer  $x_{*}$.

To enforce the global convergence, at each step we can add a line search over the space spanned by  $x_{\text{\tiny PM}}$ and
$x_{\sqp}$, where $x_{\text{\tiny PM}}$ is 
the new step generated by the  power method introduced in the above subsection. Here we show the our algorithm.
\begin{algorithm}[H]
      \caption{SQP method for  optimization problem~\eqref{goesqp:eq:98}}
      \label{combineNewton}
      \begin{algorithmic}[1]
        \Require  $\eta$: completely symmetric matrix, $\rho-\rho_{k}$ in $k$-th step in Algorithm~\ref{mainalgorithm}; $\varepsilon_{3}$: tolerance for terminate algorithm.
        \Ensure $(x_{*},f(x_{*}))$: maximizer and maximal value  of $f$ over unit sphere.
        \State Choosing initial candidate $x_0$
      
        \State $err \gets +\infty$     
      
        \State $k=0$
        \While{$err > \varepsilon_{3} $ }
              \State $k\;\:\gets k+1$
              \State $\nabla^2 f(x_k)\gets 3B_{x_{k}}$
               \State $\nabla f(x_k)\gets \nabla^2f(x_k)\, x_k/3$
              \State $\lambda_k \gets \nabla f(x_{k})^{\intercal}x_{k}$
        
              \State Compute $x_\nt$ by Eq.~\eqref{symmetrictensor:eq:30}
              \State $x_{\sqp} \gets x_\nt/ \norm{x_\nt}$ \Comment{The SQP step}
              \State $x_{\text{\tiny PM}}\gets( \nabla f(x_k)+\alpha x_{k})/\norm{\nabla f(x_k)+\alpha x_{k}} $\newline
              \Comment{Power method step, see Algorithm~\ref{PM} for choosing $\alpha$}
              \State $x_{k+1}\gets$  maximizer of $f(x)$ over the two dimensional subspace
              $\{x: x\in\mathrm{span}\{x_{\sqp},x_{\text{\tiny PM}}\},\norm{x}=1\}$
              \Comment{See Sec.~\ref{sec:modif-grad-meth} for detail}
              \State $err \gets \norm{x_k-x_{k+1}}$ 
       \EndWhile
       \State $x_{*}\gets x_{k}$
       \State $f(x_{*})\gets \bra{x_{k},x_{k}}\eta\ket{x_{k},x_{k}}/4$
      \end{algorithmic}
    \end{algorithm}
    \section{Convergence Analysis}
    \label{converge}
In this section, we investigate the    convergence behaviour of Algorithm~\ref{combineNewton}. We prove that this algorithm
converges globally and  locally quadratically.    It is convenient to assume that $f(x)$ is not a constant
during our analysis.
\begin{lem}
  \label{notzero}
  If $x$ is a maximizer of $f$ over the unit sphere, then $\nabla h(x)\neq 0$ and $x = \nabla h(x)/\norm{\nabla h(x)}$.
\end{lem}
\begin{proof}
  If $x$ is a maximizer of $f$, then by the convexity of $f$ on the unit ball, we have
  \begin{equation}
    \label{goesqp:eq:7}
    \begin{split}
    f(w)-f(x)&= h(w)-h(x)\\
    &\geqslant \nabla h(x)^{\intercal}(w-x),\;\forall \norm{w}=1.
    \end{split}
  \end{equation}
  If $\nabla h(x)=0$,  $x$ is a minimizer of $f$ over unit sphere. Hence, $f$ must be a constant over unit sphere, which contradicts with our assumption. Therefore, $\nabla h(x) \neq 0$.

  If we choose $w = \nabla h(x)/\norm{\nabla h(x)}$ in Eq.~\eqref{goesqp:eq:7}, then
  \begin{equation}
    \label{goesqp:eq:10}
    f(w) -f(x) \geqslant \norm{\nabla h(x)} (1 - w^{\intercal}x).
  \end{equation}
  If $x \neq  \nabla h(x)/\norm{\nabla h(x)}$, then
  \begin{equation}
    \label{goesqp:eq:42}
    1- w^{\intercal}x  > 0.
  \end{equation}
  Therefore, we have $f(w)> f(x)$, which contradicts to that $x$ is a maximizer.
\end{proof}
Note that  $h(x) $ can be replaced by $f(x)$ if $f(x)\neq 0$ in the above lemma.
\begin{lem}
  Let $x_{k+1}=\nabla h(x_k)/\norm{ \nabla h(x_k) } $, then $f(x_{k+1})\geqslant f(x_{k})$. Moreover if $x_k$ is not a
  KKT  point, i.e. $x_k \propto \nabla f(x_k)$, then $f(x_{k+1})> f(x_k)$.
\end{lem}
\begin{proof}
  By Taylor expansion,
  \begin{equation}
    \label{goesqp:eq:5}
    \begin{split}
      f(x) =& f(x_k) + \nabla h(x_k)^{\intercal} (x-x_k)     +\frac{1}{2} {(x-x_k)}^T \nabla^{2} h(\xi) (x-x_k),
    \end{split}
  \end{equation}
  where $\xi = tx+(1-t)x_k,0\leqslant t\leqslant 1$.  By the positivity of $\nabla^{2}h$, we have
  \begin{equation}
    \label{goesqp:eq:6}
    \begin{split}
      f(x_{k+1}) & \geqslant f(x_k) + \nabla h(x_k)^{\intercal} (x_{k+1} - x_k)\\
      & = f(x_k) + \norm{\nabla h(x_k)} (1 - x_{k+1}^{\intercal}   x_k)\\
      &\geqslant f(x_k) + \frac{1}{2}\norm{\nabla h(x_k)} \norm{x_k-x_{k+1}}^2\\
      & \geqslant f(x_k).
    \end{split}
  \end{equation}
  Note that if $x_k$ is not proportional to $x_{k+1}$, then  $\norm{\nabla f(x_k)}$ and $\norm{x_k-x_{k+1}}$ are both
  positive, which follows $f(x_{k+1})> f(x_k)$.
\end{proof}
This lemma shows that the gradient direction $\nabla h(x)$ always increase the value of function $f$ unless it is a
local  maximizer.
Note that in  Algorithm~\ref{combineNewton}, we add  a line search at each step, the new iteration will have larger function
value  compared with the power method. 
\begin{lem}%
  \label{lem1}
  The iterations $\{f(x_k)\}$ generated by Algorithm~\ref{combineNewton} is monotone increasing.
\end{lem}
Consider the case where the  algorithm terminates after finite steps. At that moment, $\nabla f(x_k)
=\pm \norm{\nabla f(x_k)}x_k$. This is a KKT point for the optimization problem~\eqref{goesqp:eq:98}. This
only implies $x_{k}$ is a critical point which could be either a local maximizer or minimizer. However, due to the round
error during the numerical computation,  $x_{k+1} = \pm x_k$ can almost never be obtained after finite steps. Hence,
hereafter in this paper,  this special case is not considered.

Note that $f(x)$ is bounded over the unit sphere, $f(x_k)$ will converge to a real value $\lambda_{*}$. The following theorem
shows that $f(x_k)$ converges to a local maximal value and $x_k$ approximates  a KKT point.
\begin{thm}
  \label{appkkt}
  Suppose $\{x_k\}$ is generated by Algorithm~\ref{combineNewton}. Then the sequence $\{f(x_k)\}$ converges. Moreover, $\{x_k,\lambda_{k}\}$
  will approximates a KKT point of optimization problem~\eqref{goesqp:eq:98}, where $\lambda_{k}=x_k^{\intercal}\nabla f(x_{k})$.
\end{thm}
\begin{proof}
  By Lamma~\ref{lem1}, $f(x_k)$ is monotone increasing. Since $f(x)$ is bounded over the unit sphere, $4f(x_k)$ converges to
  a limit point, namely $\lambda_{*}$.
  It follows that
  \begin{equation}
    \label{goesqp:eq:46}
    \lim_{k\to \infty} f(x_{k+1}) - f(x_{k}) = 0.
  \end{equation}
  By Eq.~\ref{goesqp:eq:6},
  \begin{equation}
    \label{goesqp:eq:47}
    f(x_{k+1})\geqslant f(x_k) + \frac{1}{2}\norm{\nabla h(x_k)}\norm{x_{k+1}-x_{k}}.
  \end{equation}
  Therefore,
  \begin{equation}
    \label{goesqp:eq:48}
    \lim_{k\to\infty} \norm{\nabla h(x_{k})} \norm{x_{k+1}-x_{k}} =0.
  \end{equation}
  We claim that
  \begin{equation}
    \label{g0}
    \lowlim_{k\to\infty}\norm{\nabla h(x_k)}>0.
      \end{equation}
  Suppose otherwise  $\lowlim\limits_{k\to\infty}\norm{h(x_k)}=0$, then there exists a subsequence
  $\{x_{k_{l}}\}$ such that
  \begin{equation}
    \label{goesqp:eq:49}
    \lim_{l\to \infty} \norm{\nabla h(x_{k_{l}})} = 0.
  \end{equation}
  Now that $\{x_{k_{l}}\}$ is a bounded sequence,  it again contains a convergent subsequence $\{x_{k_{l_{p}}}\}$ such
  that
  \begin{equation}
    \label{goesqp:eq:50}
    \lim_{p\to\infty} x_{k_{l_{p}}} = x_{*}.
  \end{equation}
  Hence,
  \begin{equation}
    \label{goesqp:eq:51}
    \lim_{p\to \infty} \norm{\nabla h(x_{k_{l_{p}}})} = \norm{\nabla h(x_{*})} = 0.
  \end{equation}
  Since $f(x_*)$  is a maximizer among $\{x_{k}\}$, Eq.~\ref{goesqp:eq:51}
   contradicts with Eq.~\eqref{notzero}. Therefore, Eq.~\eqref{g0} holds.
   By Eq.~\ref{goesqp:eq:47}, we have
   \begin{equation}
     \label{goesqp:eq:52}
     \lim_{k\to\infty}x_{k+1}-x_k = 0.
   \end{equation}
   Let $p_{k+1} = \nabla h(x_k)/\norm{\nabla h(x_{k})}$. Similarly, by Eq.~\eqref{goesqp:eq:6}, we have
   \begin{equation}
     \label{goesqp:eq:56}
     \lim_{k\to\infty}p_{k+1} - x_k =0.
   \end{equation}
  Let $\epsilon_k = p_{k+1}-x_{k}$, we have
  \begin{equation}
    \label{goesqp:eq:53}
    \begin{split}
      \norm{\nabla h(x_k)} &= \nabla h(x_k) ^{\intercal} (x_{k} + \epsilon_k)\\
      & = \nabla h(x_k)^{\intercal}x_k +  \nabla h(x_k)^{\intercal}\epsilon_k.
      \end{split}
    \end{equation}
    Therefore,
    \begin{equation}
      \label{goesqp:eq:54}
      \begin{split}
        \lim_{k\to\infty} \norm{\nabla h(x_k)}& =\lim_{k\to\infty} \nabla h(x_k)^{\intercal}x_k\\
        & = \lim_{k\to\infty}(4f(x_k) + \alpha)\\
        & = \lambda_* +\alpha.
        \end{split}
   \end{equation}
   Let $\lambda_{k} = \nabla f(x_{k})^{\intercal}x_{k}$, then
   \begin{equation}
     \label{goesqp:eq:55}
     \begin{split}
   \lim_{k\to\infty} \nabla f(x_k) - \lambda_{k} x_k
       & =\lim_{k\to\infty}\norm{\nabla h(x_k)} x_{k} -(\nabla f(x_{k})^{\intercal}x_{k}+\alpha)x_{k}\\
       &=\lim_{k\to\infty}\norm{\nabla h(x_k)}( x_{k}+\epsilon_k) -(4 f(x_{k})+\alpha)x_{k}\\
     & = \lim_{k\to\infty} (\norm{\nabla h(x_{k}) } - \lambda_{*}-\alpha)x_k\\
     & = 0.
     \end{split}
   \end{equation}
\end{proof}

Note that Theorem~\ref{appkkt} shows that the sequence $\{\lambda_k\}$ produced by Algorithm~\ref{combineNewton} converges to a local
maximum of $f$,
but $\{x_k,\lambda_k\}$ only approximates a KKT point of the optimization problem, which means $\{x_{k}\}$ may be not
convergent.
Furthermore, we show it converges if there are only finite many KKT points.
\begin{thm}
  Let $\{x_{k}\}$ be the sequence  produced by Algorithm~\ref{combineNewton}. If there are only finite many KKT points for the
  optimization 
  problem~\eqref{goesqp:eq:98}, then $\{x_{k}\}$ converges.
\end{thm}
\begin{proof}
  By Theorem~\ref{appkkt}, any accumulation point $\{x_{k}\}$ must be a KKT point for the problem
  ~\eqref{goesqp:eq:98}.
  Suppose there exist only finite many accumulates of $\{x_k\}$, namely $\bar x_{t},t=1,2,\ldots,T$.

  Let
  \[\epsilon_d = \frac{1}{3} \min_{1\leqslant t_1 < t_2\leqslant T}\norm{\bar x_{t_1}-\bar x_{t_2}},\]
  and
  \[B_t = \{x:\norm{x-\bar x_t}< \epsilon_d, t=1,2,\ldots,T\}.\] Then
  \begin{equation}
    \label{goesqp:eq:57}
    \norm{x-y}> \epsilon_d, \text{ if } x\in B_{t_{1}},y\in B_{t_2},  t_1\neq t_{2}.
  \end{equation}
  Note that there are only finite many $x_{k}$ such that $x_{k}\not\in \cup_{t=1}^{T} B_{t}$, which means there exists a
  number $K_1>0$ such that 
  \begin{equation}
    \label{goesqp:eq:58}
    x_k \in \bigcup_{t=1}^{T}B_{t},\;\forall k>K_{1}.
  \end{equation}
  Since $\{x_{k}\}$ is bounded, there exist a subsequence $x_{k_{l}}$ such that it converges to an accumulation point,
  namely $\bar x_{1}$. Hence, there exists a number $K_{2}>0$ such that
  \begin{equation}
    \label{goesqp:eq:59}
    x_{k_{l}}\in B_{1},\; \forall l>K_2.
  \end{equation}
  By Eq.~\eqref{goesqp:eq:52}, there exists a number $K_{3}>0$ such that
  \begin{equation}
    \label{goesqp:eq:60}
    \norm{x_{k+1}-x_{k}}< \epsilon_d,\; \forall k>K_{3}.
  \end{equation}
  Let $K = \max\{K_1,k_{K_2},K_3\}$.
  For any $k>K$, by Eq.~\eqref{goesqp:eq:59}, there exists a  $k_{l}>K$ such that $x_{k_{l}}\in B_{1}$. We claim that $x_{k_{l}-n}\in B_{1}$ for any
  $k_{l}-n>K$.

  By Eqs.~\eqref{goesqp:eq:57},~\eqref{goesqp:eq:58}, and~\eqref{goesqp:eq:60}, we have
   $x_{k_{l}-1}\in$ if $k_{l}-1>K$. Further, by the mathematical induction,  $x_{k_{l}-n}\in B_{1}$ for any
  $k_{l}-n>K$. If choose $n = k_{l}-k$, then we have
  \begin{equation}
    \label{goesqp:eq:61}
    x_k \in B_1,\; \forall k > K.
  \end{equation}
  It follows that $\bar x_1$ is the only one accumulation point of $\{x_k\}$ and thus it converges.
\end{proof}
Note that any KKT point $\bar x$ is also an eigenvector of $B_{\bar x}$. Therefore, the following condition can guarantee
that there are only  finite many  accumulation points of $\{x_k\}$.
\begin{corollary}
  If the eigenvalues of $B_{x}$ are simple for any $x$, then $\{x_k\}$ generated by Algorithm~\ref{combineNewton} converges.
\end{corollary}
We have proved the global convergence of $\{x_{k}\}$ produced by Algorithm~\ref{combineNewton}. Next,  we finish this section by
proving the locally quadratic convergence.
\begin{thm}
  Let $x_{*}$ be a local maximizer of the optimization problem~\eqref{goesqp:eq:98} at which the second order sufficient conditions are
  satisfied. Then
  Algorithm~\ref{combineNewton} converges quadratically near $x_{*}$.
\end{thm}
\begin{proof}
  By the standard convergence theorem for Newton's method applied to a smooth system of equations, there exists an $\varepsilon>0$,
  and a $q>0$ such that for $\norm{x_{k} - x_{*}} +|\lambda_{k}-\lambda_{*}|\leqslant \varepsilon $,
  \begin{equation}
    \label{goesqp:eq:8}
    \begin{split}
      \norm{x_\nt - x_{*}} +|\lambda_{\nt}-\lambda_{*}| 
     \leqslant q(\norm{x_k-x_{*}}^2 + |\lambda_k - \lambda_{*}|^2).
    \end{split}
  \end{equation}
   On the other hand by Lemma~\ref{lem7},
  \begin{equation}
    \label{goesqp:eq:13}
\abs{    f(x_{*}) - f(x)} < a_2 \norm{x-x_{*}}^2.
\end{equation}
Note that $\lambda_{k} = \nabla_f(x_k)^{\intercal} x_k $, then
  \begin{equation}
    \label{goesqp:eq:16}
    \begin{split}
      |    \lambda_{k} - \lambda_{*}| &=  | x_k^{\intercal}B(x_k)x_k  -   x_{*}^{\intercal}B(x_{*})x_{*} |\\
     &=  4 |f(x_{*}) - f(x_k)|\\
      &\leqslant 4 a_2\norm{x_{*} - x_k}^2
    \end{split}
  \end{equation}
  Hence, Eq.~\eqref{goesqp:eq:8} holds if
  \begin{equation}
    \label{goesqp:eq:128}
    \norm{x_{k}-x_{*}}+4a_2\norm{x-x_{*}}^{2} \leqslant \varepsilon.
  \end{equation}
  Assume $\varepsilon <1$, then $\norm{x_{k}-x_{*}}<1$. Forward, Eq.~\ref{goesqp:eq:8} satisfies if
  \[\norm{x_k-x_{*}}\leqslant
  \frac{\varepsilon}{1+4a_2}.\]

Moreover, by Eq.~\eqref{goesqp:eq:8} can be reformulated as
\begin{equation}
  \label{goesqp:eq:129}
  \begin{split}
    \norm{x_\nt-x_{*}} \leqslant &(q+\frac{  a_2^{2}}{{ (1+4a_2)}^{2} }   \varepsilon )\norm{x_k-x_{*}}^{2}\\
    \leqslant &r\norm{x_{k}-x_{*}}^{2},
  \end{split}
\end{equation}
where $r = q+\frac{a_2^{2}}{{(1+4a_2)^2}}\varepsilon $.

We thus can get rid of $\lambda_{k}$ in our analysis.  Consider the Lagrange function
$\Lag(x,\lambda) = f(x) -\frac{\lambda}{2}( x^\intercal x-1)$, we have
  \begin{equation}
    \label{goesqp:eq:9}
    f(x) = f(x_{*}) + \frac{1}{2}(x-x_{*})^{\intercal}(\nabla^2f(\xi)-\lambda_{*}\I)(x-x_{*}),
  \end{equation}
  where $\xi = tx +(1-t)x_{*}$, $0<t<1$.

  By Lemma~\ref{lem4}, we have
  \begin{equation}
    \label{goesqp:eq:11}
    \begin{split}
     \abs{ f(x_{*}) - f(x)} 
      &=\frac{1}{2}\abs{(x-x_{*})^{\intercal}(\nabla^2f(\xi)-\lambda_{*}\I)(x-x_{*})}\\
      &\geqslant  \frac{a_1}{2} \norm{x-x_{*}}^2 - \frac{1}{2}\abs{(x-x_{*})^{\intercal}(\nabla^2 f(x_{*}) -\nabla^{2}f(\xi))(x-x_{*})}\\
      & \geqslant  \frac{1}{2}\norm{x-x_{*}}^2 (a_{1} - \norm{\nabla^{2} f(x_{*}) -\nabla^{2} f(\xi)}_{2} ).
    \end{split}
  \end{equation}
  By the Lipschitz condition~\eqref{goesqp:eq:124},
  \begin{equation}
    \label{goesqp:eq:127}
    \begin{split}
    \abs{ f(x_{*}) - f(x)}& \geqslant \frac{1}{2}(a_{1}-6\norm{\eta}_{2}\norm{x-\xi})\norm{x-x_{*}}^{2}\\
  &=  \frac{1}{2}(a_{1}-6(1-t)\norm{\eta}_{2}\norm{x-x_{*}})\norm{x-x_{*}}^{2}\\
&\geqslant  \frac{1}{2}(a_{1}-6\norm{\eta}_{2}\norm{x-x_{*}})\norm{x-x_{*}}^{2}.
  \end{split}
  \end{equation}
  When $x$ is close to $x_{*}$, say $\norm{x-x_{*}}< \frac{1}{12\norm{\eta}_{2}}$,
  \begin{equation}
    \abs{f(x_{*}) - f(x)} >  \frac{a_1}{4} \norm{x-x_{*}}^2.
  \end{equation}
  We can further assume that $f(x_{*})\geqslant f(x)$ when $x$ is close to $x_{*}$. Hence,
    \begin{equation}
    \label{goesqp:eq:12}
   f(x_{*}) - f(x) >  \frac{a_1}{4} \norm{x-x_{*}}^2.
  \end{equation}

  By Eq.~\ref{goesqp:eq:117}
  \begin{equation}
    x_{\sqp} = \frac{x_\nt}{\norm{x_\nt}}.
  \end{equation}
      Note that in Eq.~\ref{goesqp:eq:130}, $p_k^{\intercal}x_{k}=0$, then $\norm{x_\nt}=\norm{x_{k}+p_{k}}>1$. Therefore,
      \begin{equation}
        \label{xsqp}\norm{x_{\sqp} - x_{*}} \leqslant \norm{x_\nt - x_{*}}.
      \end{equation}
      For simplicity, the inequality Eq.~\eqref{xsqp} is illustrated in appendix \ref{apdxsqp}.
      
      Note that $f(x_{k+1})\geqslant f(x_\sqp)$, by Eq.~\ref{goesqp:eq:12}  we have
      \begin{equation}
    \label{goesqp:eq:18}
    \begin{split}
      \frac{a_1}{4} \norm{x_{k+1}-x_{*}}^{2}&\leqslant f(x_{*}) - f(x_{k+1}) \\
      & \leqslant f(x_{*}) - f(x_\sqp)\\
      & \leqslant a_2\norm{x_{\sqp}-x_{*}}^{2}\\
      &\leqslant a_2r^2\norm{x_{k}-x_{*}}^2.
    \end{split}
  \end{equation}
  Forward,, we have
  \begin{equation}
    \label{goesqp:eq:19}
    \begin{split}
      \norm{x_{k+1}-x_{*}}\leqslant
      2r\sqrt{\frac{a_2}{a_1}} \norm{x_k - x_*}^2.
    \end{split}
  \end{equation}
  Therefore, $x_k$ converges quadratically near $x_{*}$.
\end{proof}

\section{Numerical examples }
\label{numerical}
~In this section, we implement some numerical experiments. We use Algorithm~\ref{combineNewton} to solve the
sub-problem~\eqref{goesqp:eq:98} per iteration in Algorithm~\ref{mainalgorithm} and compare the convergence behaviour  on the following examples: 
\begin{example}
  \label{ex1}
  The completely symmetric state $\eta$ generated by
  \begin{equation}
    \label{goesqp:eq:1}
    \eta = \sum_{i=0}^{L-1}p_i\ket{x_i,x_i}\bra{x_i,x_i},
  \end{equation}
  where $L>N$, $p_i  >0$, and  $x_i\in\real^{N},\norm{x_i}=1,i=0,1,\ldots,L$. And $L$ is chosen to be $4N$ in the
  numerical experiments.
\end{example}

\begin{example}
  \label{ex2}
  The completely symmetric matrix, which may not be positive, defined by
  \begin{equation}
    \label{goesqp:eq:2}
    \eta = \sum_{i=0}^{L-1} p_i\ket{x_i,x_{i}}\bra{x_{i},x_i},
  \end{equation}
  where $L>N$, $p_{i}\in[-1,1]$, and $x_i\in\real^{N},\norm{x_i}=1,i=0,1,\ldots L$. And $L$ is chosen to be $4N$ in the
  numerical experiments.
\end{example}
\begin{example}
  \label{ex3}
  The completely symmetric matrix $\eta$ defined by
  \begin{equation}
    \label{S-sep-v7:eq:1}
    \eta = \sum_{i,j,k,l=0}^{N-1}\eta_{ijkl}\ket{i,k}\bra{j,l},
  \end{equation}
  where 
  \begin{equation}
    \label{goesqp:eq:15}
    \eta_{ijkl} = \frac{1}{4N^{3}}(i+j+k+l),\;0\leqslant i,j,k,l\leqslant N-1.
  \end{equation}
  
\end{example}
\begin{example}
~\label{ex4}
  The completely symmetric matrix $\eta$ defined by
   \begin{equation}
    \label{S-sep-v7:eq:01}
    \eta = \sum_{i,j,k,l=0}^{N-1}\eta_{ijkl}\ket{i,k}\bra{j,l},
  \end{equation}
  where 
  \begin{equation}
    \label{goesqp:eq:3}
    \eta_{ijkl} = \frac{1}{N^{2}}\sin(\frac{i+j+k+l}{4N}),\;0\leqslant i,j,k,l\leqslant N-1.
  \end{equation}
\end{example}

 We test these examples for $N = 4,8,\ldots,512$, where $N$ is the dimension of $x$.
The examples are tested on the computer with 12-core CPU
2.40GHz   
and they are implemented on the MATLAB version 9.4.0.813654 (R2018a).

Some parameters of Algorithm~\ref{combineNewton}  are set as follows:
\begin{table}[H]
  \centering
  \begin{tabularx}{0.9\linewidth}{XXX}
    \toprule
      Variables &  Value & Meaning \\\midrule
      $\epsilon_{3}$ &  $10^{-12}$ & Tolerance for termination \\
      $N_{\max}$& 500& Maximal iteration number\\
       $x_0$&  rand(N,1)& Initial  vector\\\bottomrule
    \end{tabularx}
  \caption{Parameters in algorithms~\ref{combineNewton}}
\end{table}
The following figures show the global and local convergence of Algorithm~\ref{combineNewton}.
\vskip-20pt
  \noindent
     \begin{figure}[H]
       \begin{subfigure}{0.495\linewidth}
    \includegraphics[width=\linewidth]{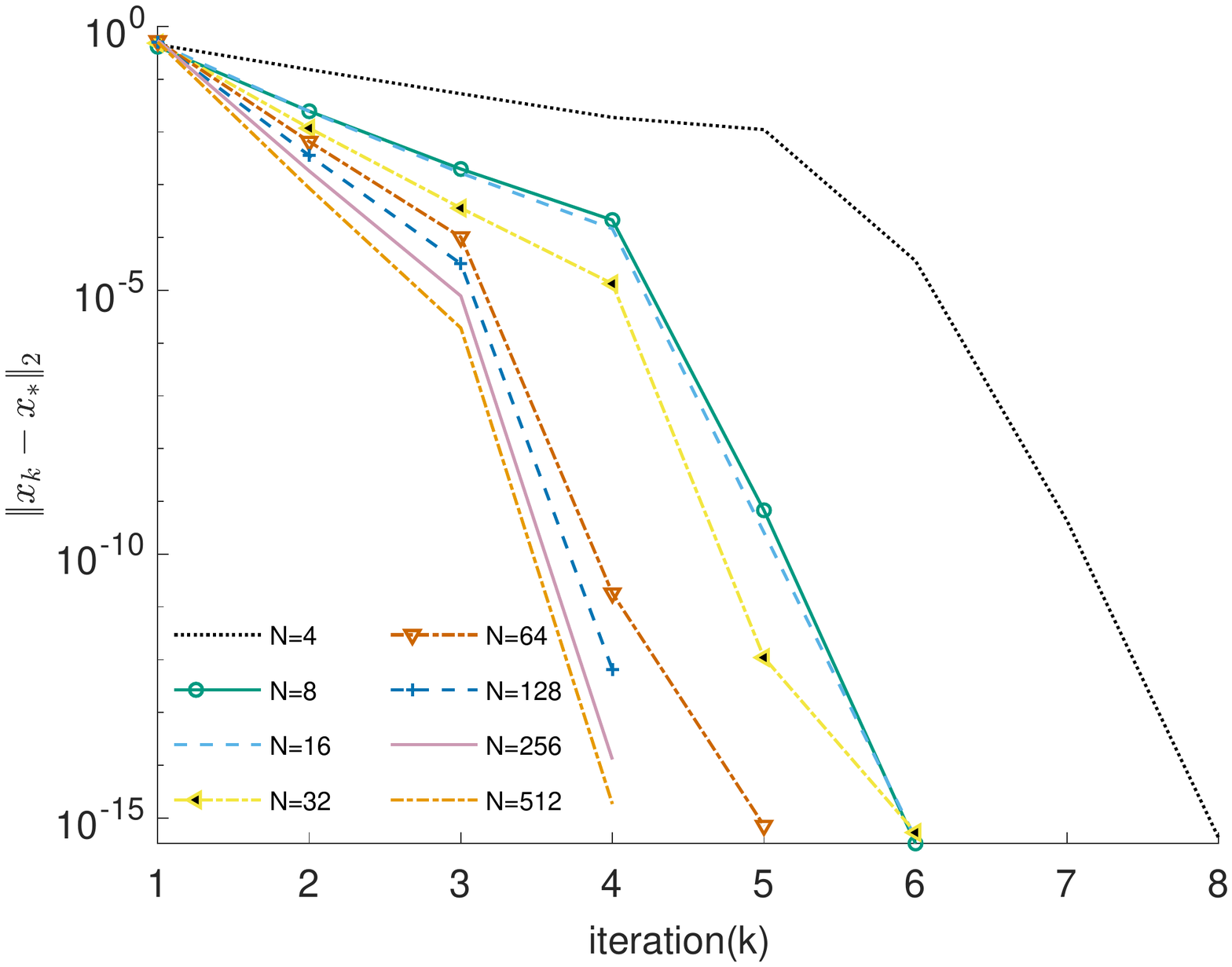}
    \caption{ Convergence of $x_{k}$}
  \end{subfigure}
  \begin{subfigure}{0.495\linewidth}
    \includegraphics[width=\linewidth]{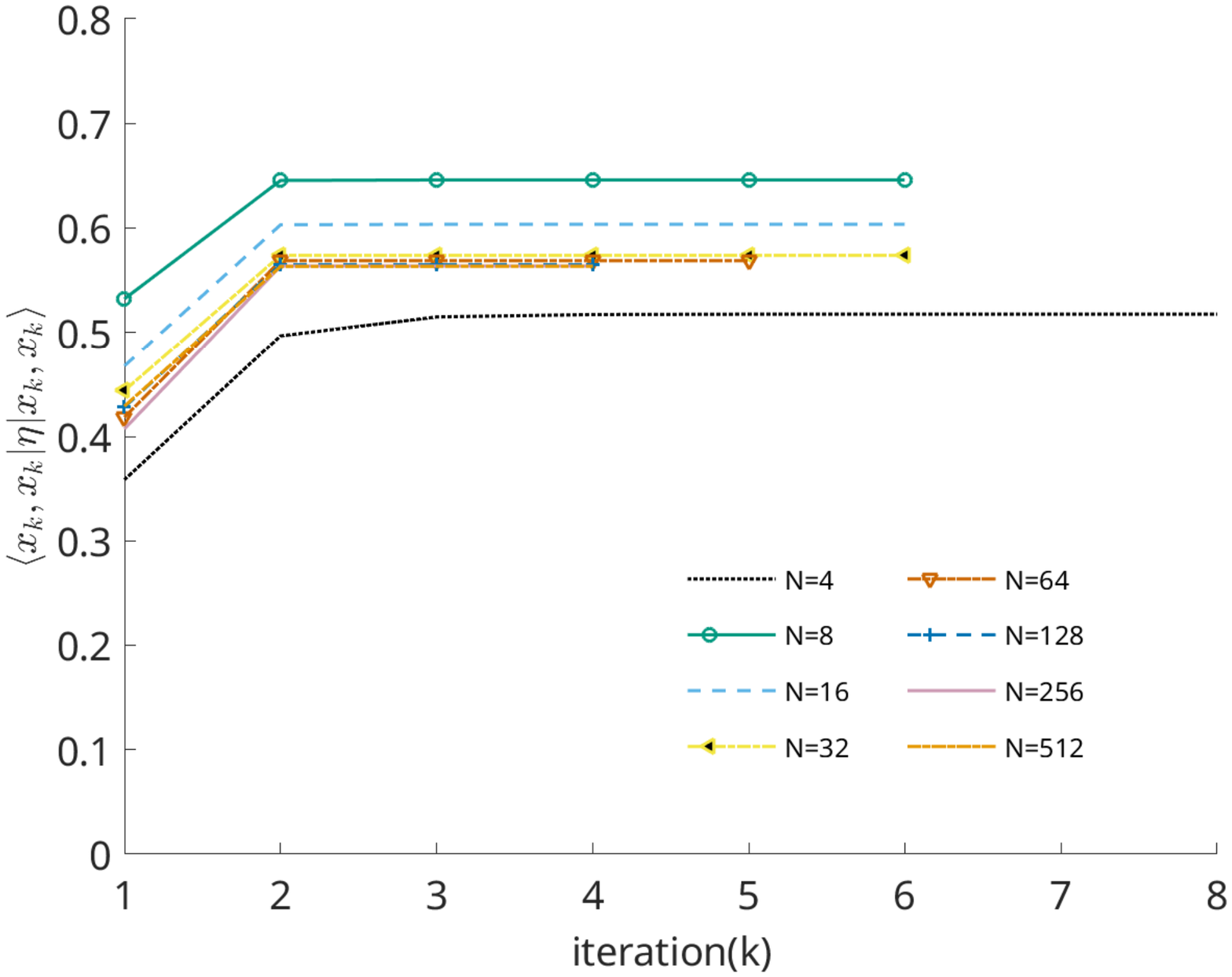}
    \caption{Monotone increasing of $\bra{x_{k},x_{k}}\eta\ket{x_k,x_{k}}$}
   \end{subfigure}
   \caption{Convergence behaviour of Alg.~\ref{combineNewton} for Example~\ref{ex1}}
 \end{figure}
\vskip-70pt
 \begin{figure}[H]
   \centering
   \begin{subfigure}{0.495\linewidth}
           \includegraphics[width=\linewidth]{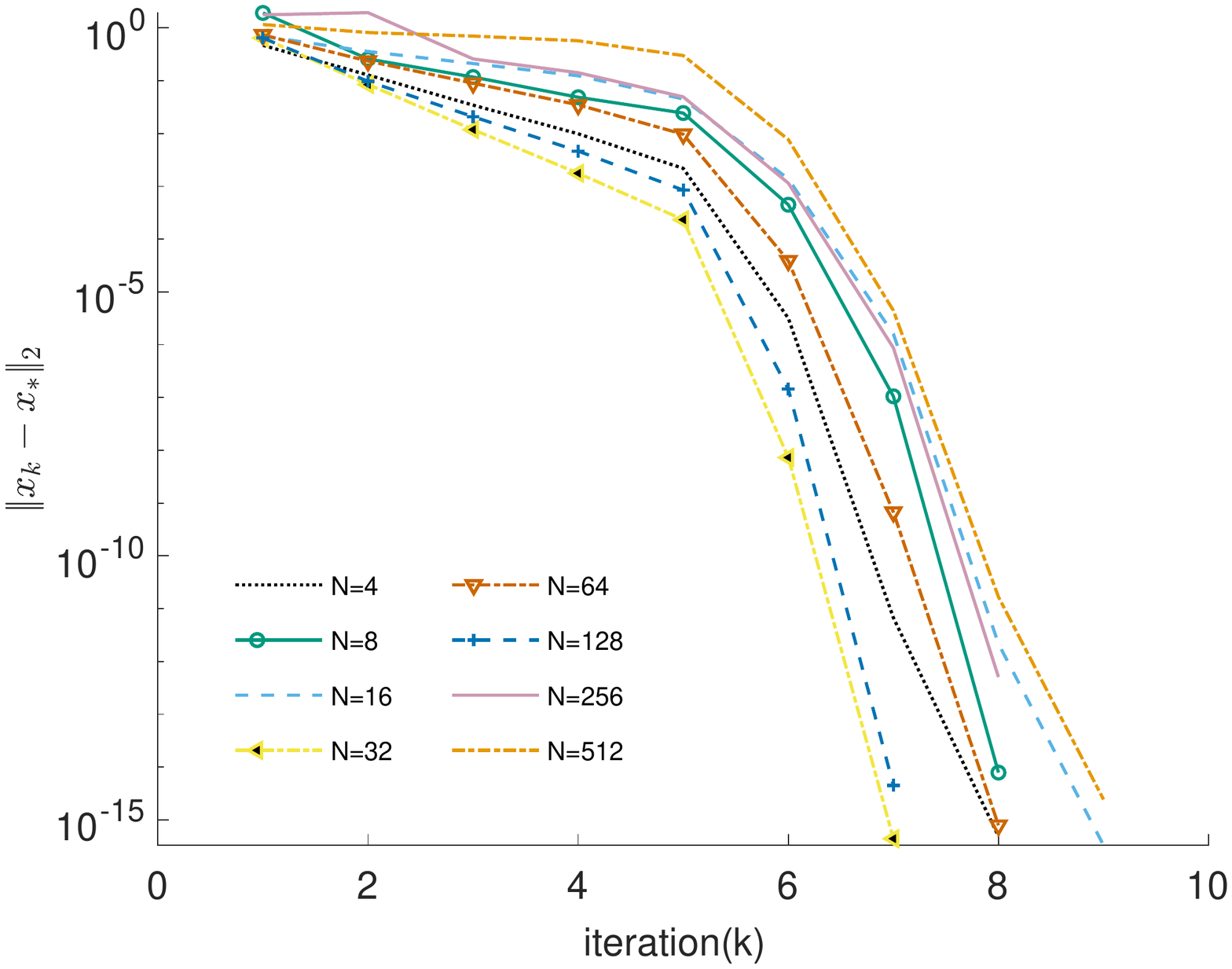}
      \caption{   Convergence of $x_{k}$ }

    \end{subfigure}
    \begin{subfigure}{0.495\linewidth}
            \includegraphics[width=\linewidth]{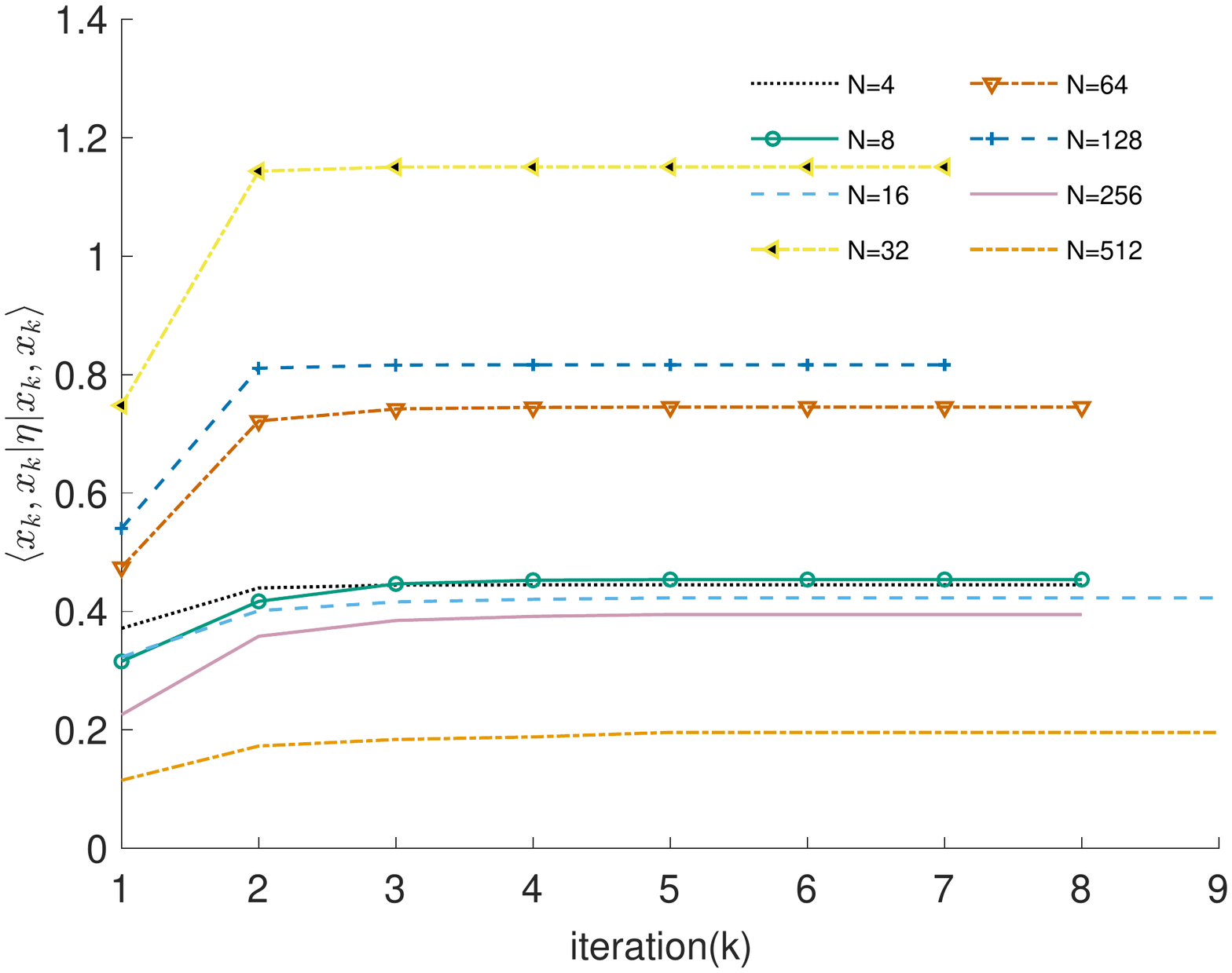}   
    \caption{Monotone increasing of $\bra{x_{k},x_{k}}\eta\ket{x_k,x_{k}}$}
    \end{subfigure}
   \caption{Convergence behaviour of Alg.~\ref{combineNewton} for Example~\ref{ex2}}
    \end{figure}\vskip-102pt

    \clearpage
    \noindent
    ~
    \vskip-60pt
    \begin{figure}[H]
      \label{figex3}
    \centering
    \begin{subfigure}{0.495\linewidth}
      \includegraphics[width=\linewidth]{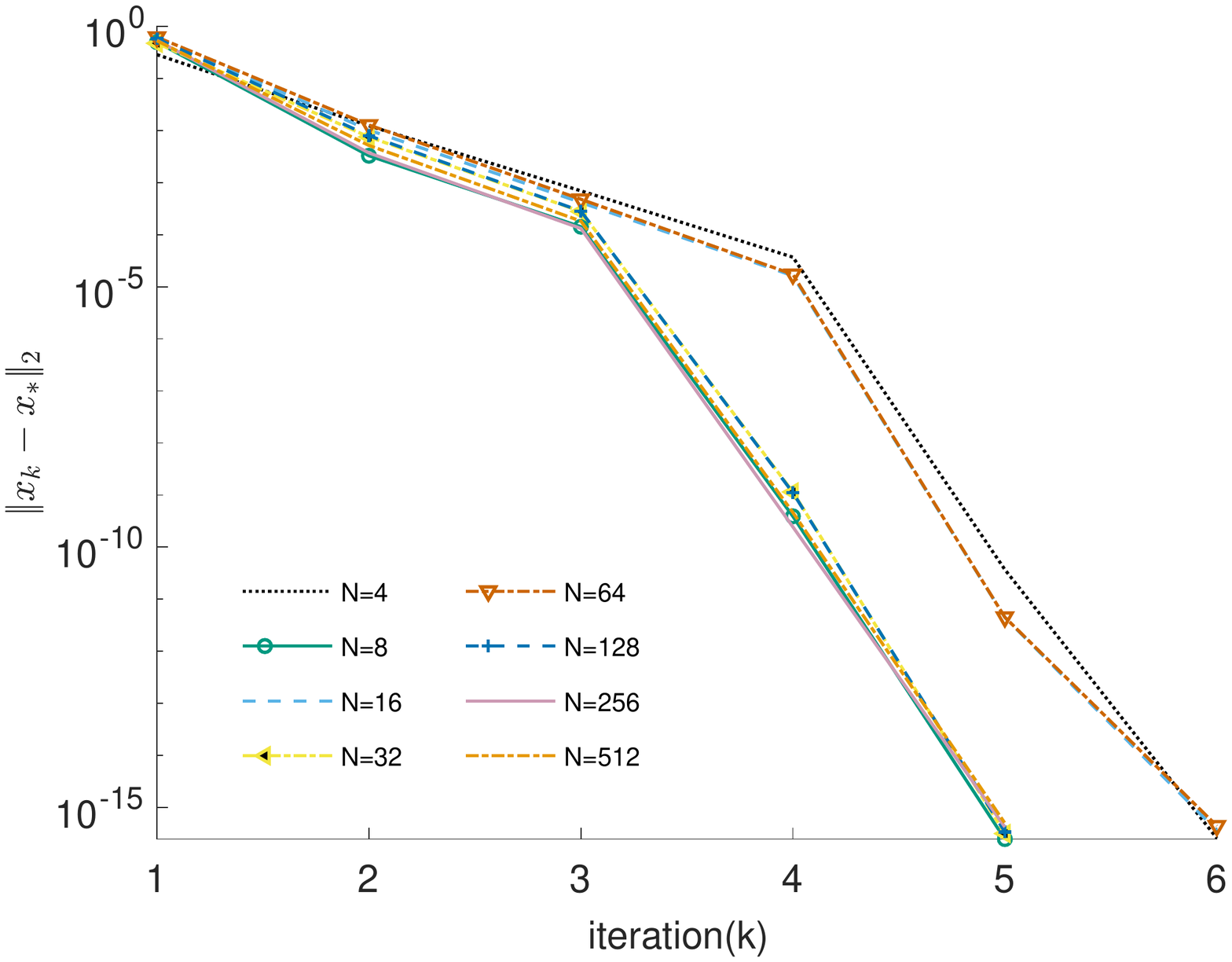}
      \caption{ Convergence of $x_{k}$}
    \end{subfigure}
    \begin{subfigure}{0.495\linewidth}
      \includegraphics[width=\linewidth]{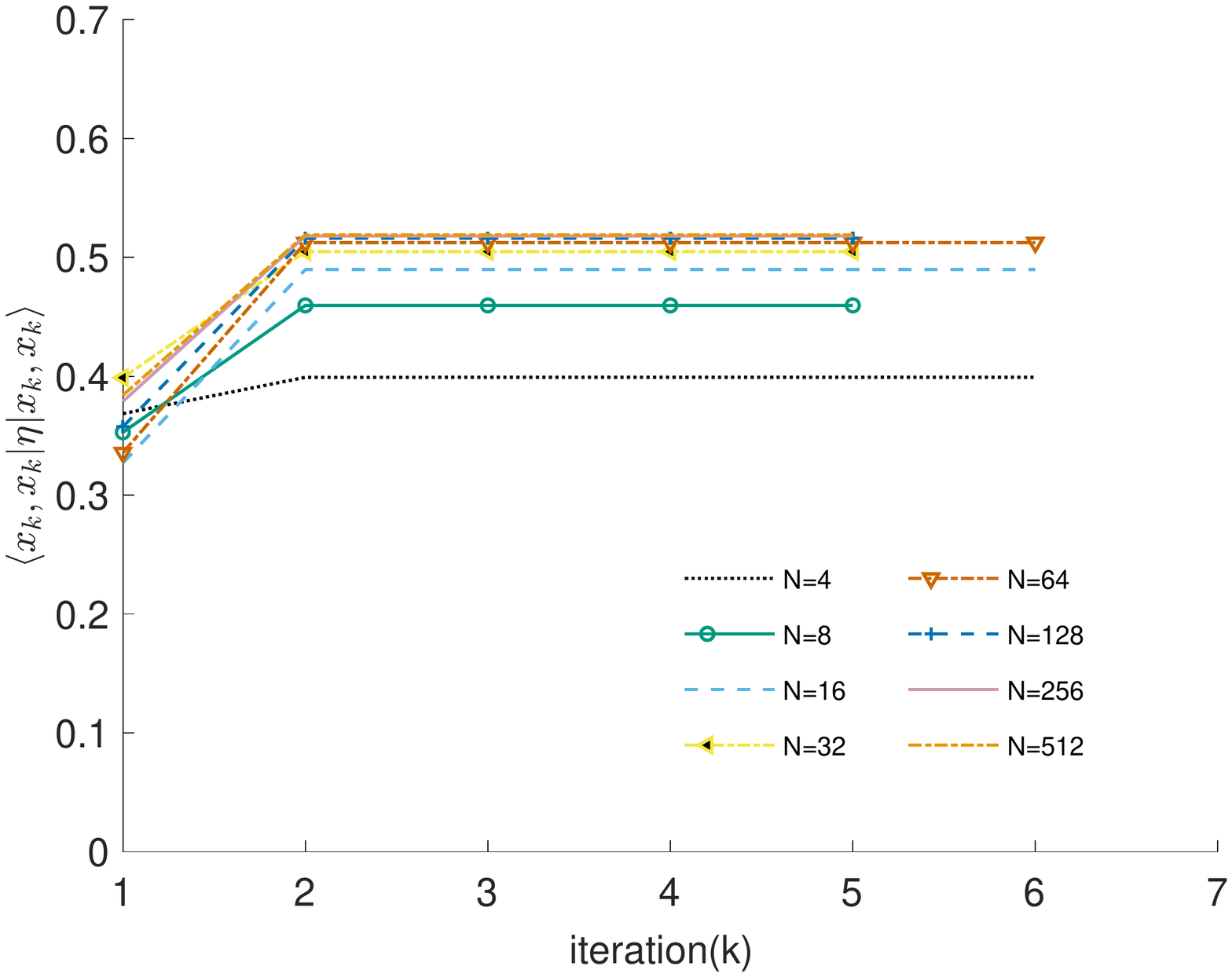}
      \caption{Monotone increasing of $\bra{x_{k},x_{k}}\eta\ket{x_k,x_{k}}$}
      \end{subfigure}
    \caption{Convergence behaviour of Alg.~\ref{combineNewton} for Example~\ref{ex3}}
  \end{figure}
  \vskip-60pt
  \begin{figure}[H]
    \centering
    \begin{subfigure}{0.495\linewidth}
      \includegraphics[width=\linewidth]{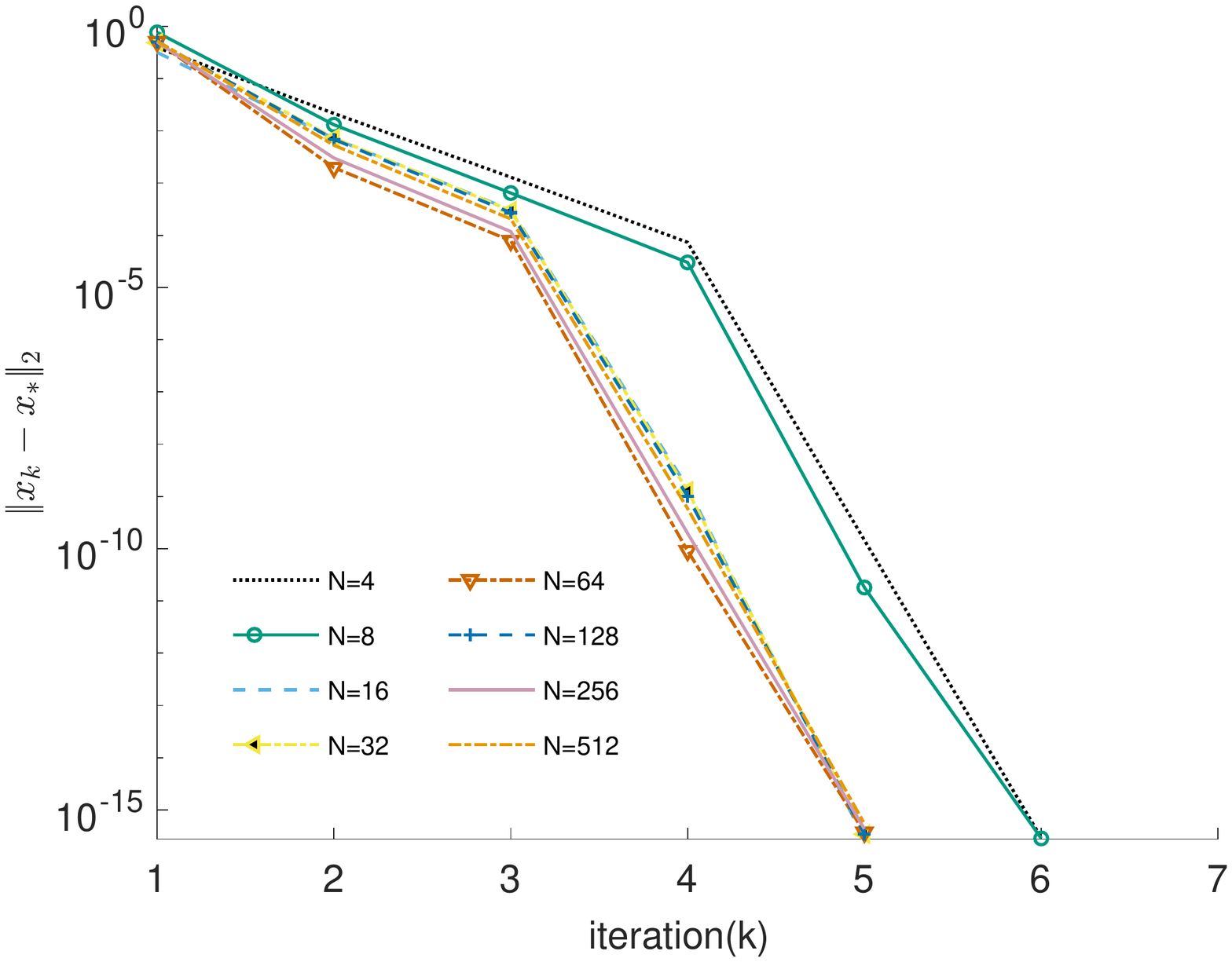}
      \caption{ Convergence of $x_{k}$}
      \end{subfigure}
    \begin{subfigure}{0.495\linewidth}
      \includegraphics[width=\linewidth]{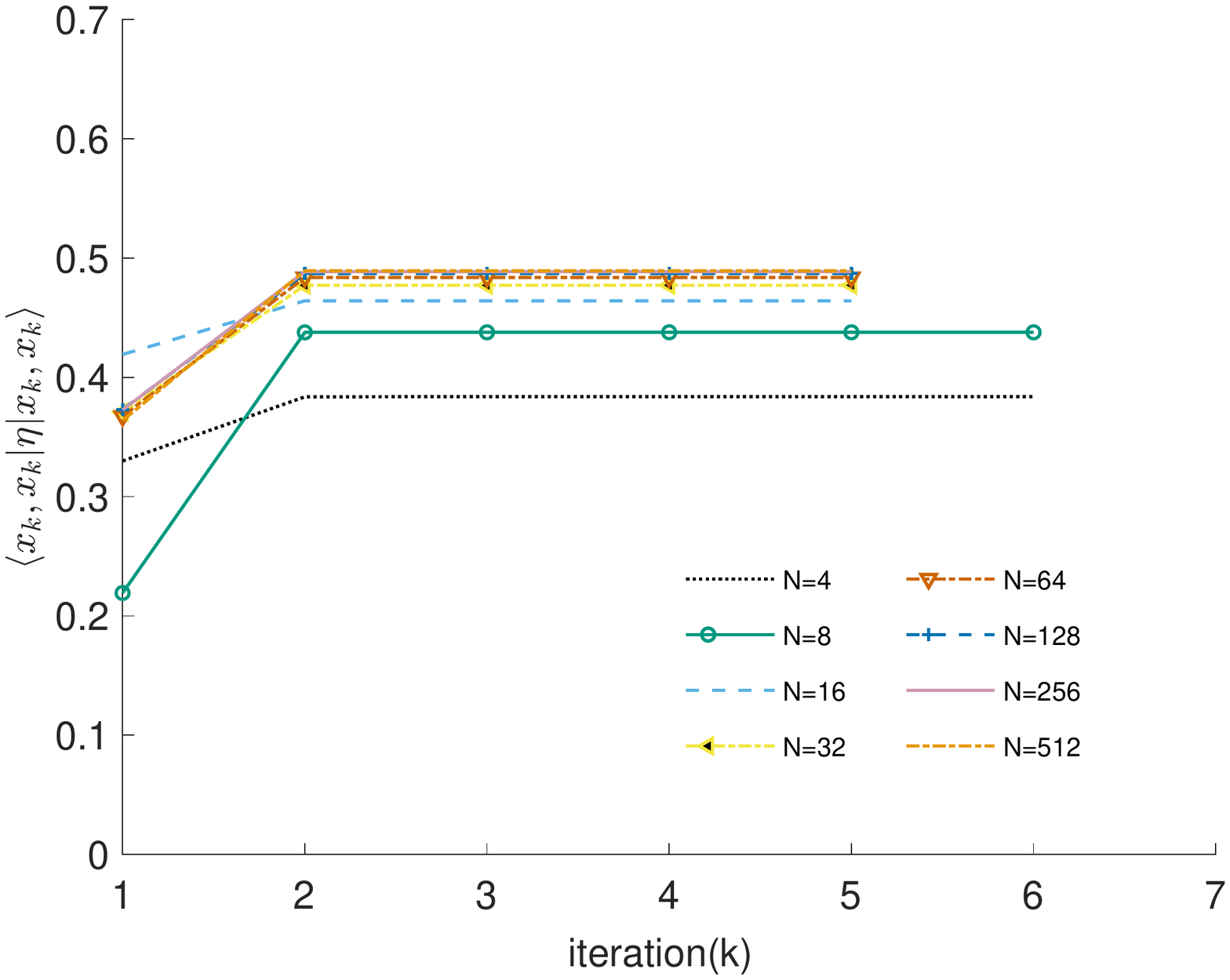}
      \caption{Monotone increasing of $\bra{x_{k},x_{k}}\eta\ket{x_k,x_{k}}$}
    \end{subfigure}
    \caption{Convergence behaviour of Alg.~\ref{combineNewton}  for Example~\ref{ex4}}
  \end{figure}
 
 From these figures, we can see that the Algorithm~\ref{combineNewton} converges globally and locally quadratically. And the
 iteration number of Algorithm~\ref{combineNewton} will not increase with the dimension of matrix $\eta$.  Besides,
 $\bra{x_{k},x_{k}}\eta\ket{x_k,x_{k}}$ is monotonically increasing, which converges to a local maximum of the
 optimization problem~\eqref{goesqp:eq:98}. In practice, in  order to save the memory when the matrix size is very
 large,
 we don not
 generate $\eta$. For Example \ref{ex1} and \ref{ex2}, only $\lambda_{i},x_{i}$ is saved. Then the gradient vector and Hessian matrix can
 be calculated directly by these gradients. For Example \ref{ex3} and \ref{ex4}, we have explicit formula to calculate the gradient vector
 and Hessian matrix.
 One problem in this algorithm is that we cannot guarantee that the local
 maximum found is the  global maximum. It requires further research to propose a global optimization algorithm.

 We also test Algorithm~\ref{mainalgorithm} by choosing $\rho$ as $\eta$ as  in Example \ref{ex1},\ref{ex2},\ref{ex3}, and \ref{ex4}.
  Note that, in Example~\ref{ex1}, the closest S-separable state of $\rho$ 
  is $\rho$ itself, the distance $\norm{\rho - \proj(\rho)}_{F}$ is thus $0$. For Example \ref{ex2},\ref{ex3}, abd \ref{ex4}, the distance
  $\norm{\rho-\proj(\rho)}_{F}$ is not 0, which implies that these states are not S-separable.
  We compare the different results 
  for $N=4,8,16,\ldots,512$.

Some parameters for Algorithm~\ref{mainalgorithm} are set as follows:
\begin{table}[H]
  \centering
  \begin{tabularx}{0.9\linewidth}{XXX}
    \toprule
      Variables &  Value & Meaning \\\midrule
      $\epsilon_1$ &  $10^{-12}$ & Tolerance in Alg.~\ref{mainalgorithm} \\
      $\epsilon_3$ & $10^{-12}$ & Tolerance in Alg.~\ref{combineNewton}. \\
      $N_{\max}$& 1000& Maximal iteration number in Alg.~\ref{mainalgorithm}\\
      $\tilde N_{\max}$ & 500 & Maximal iteration number in Alg.~\ref{combineNewton}\\
      $x_0$&  rand(N,1)& Initial  vector in Alg.~\ref{combineNewton}\\
      $\rho_{0}$& zeros(N*N) & Initial candidate  in Alg.~\ref{mainalgorithm}\\\bottomrule
    \end{tabularx}
  \caption{Parameters in Algorithm~\ref{mainalgorithm}}
\end{table}

  ~
 \vskip-80pt\noindent
\begin{minipage}[t]{0.5\linewidth}
     \begin{figure}[H]
        \includegraphics[width=0.9\linewidth]{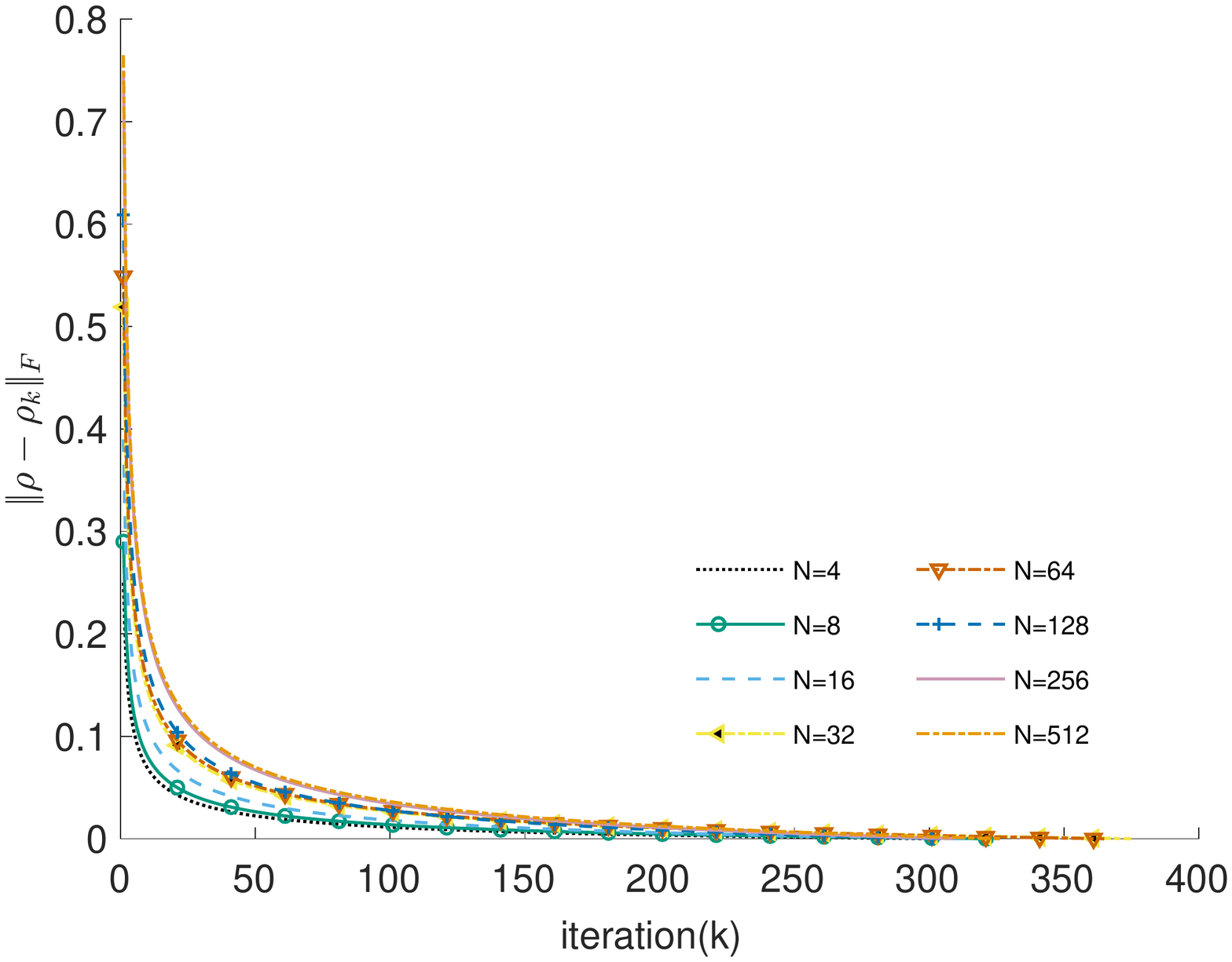}
 \caption{$\norm{\rho-\rho_{k}}_{F}$ in Example~\ref{ex1}}
\end{figure}
 \end{minipage}
  \begin{minipage}[t]{0.5\linewidth}
     \begin{figure}[H]
    \includegraphics[width=0.9\linewidth]{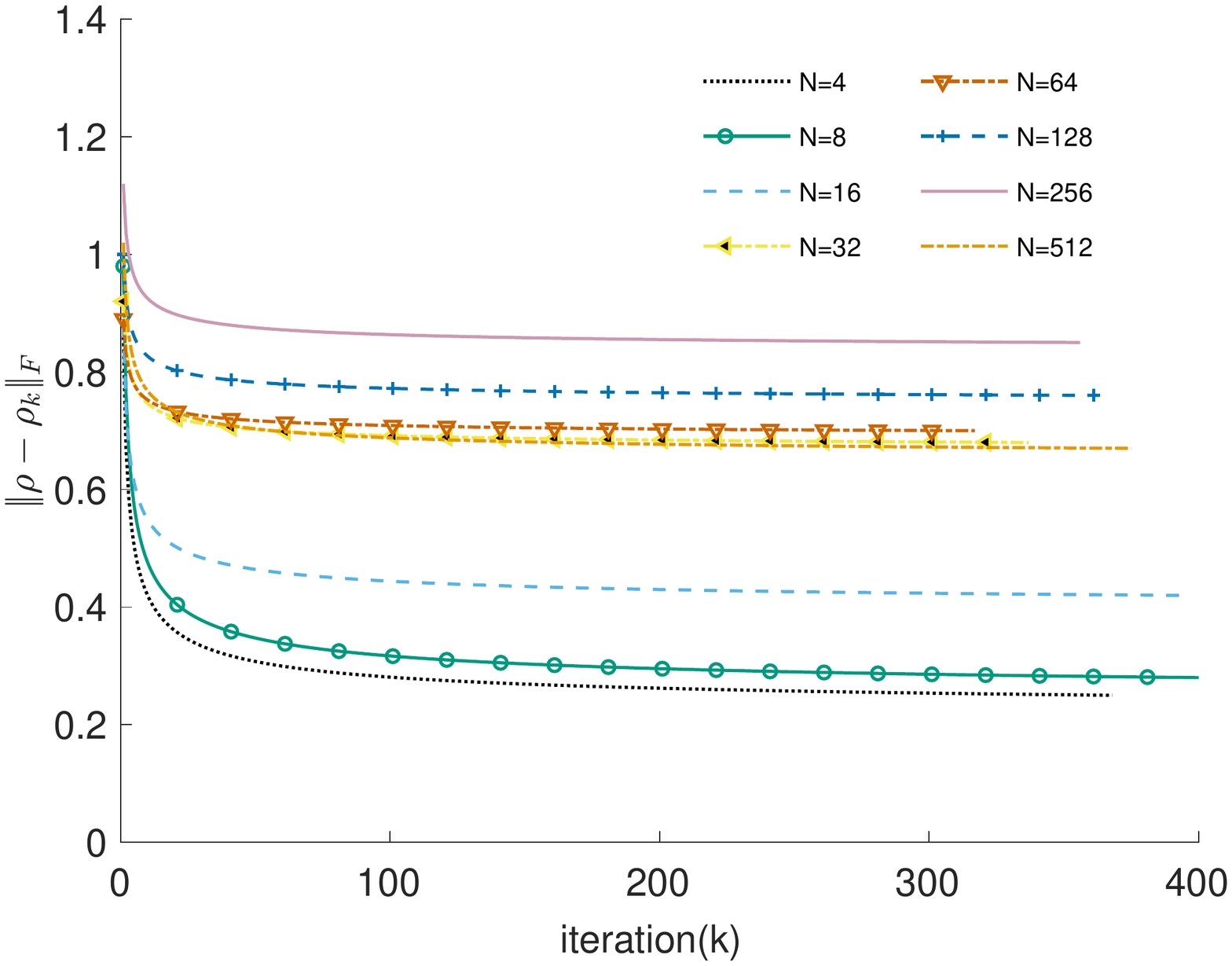}
   \caption{$\norm{\rho-\rho_{k}}_{F}$ in Example~\ref{ex2} }
 \end{figure}
 \end{minipage}
\vskip-60pt \noindent
  \begin{minipage}[t]{0.5\linewidth}
    \begin{figure}[H]
      \centering
      \includegraphics[width=0.9\linewidth]{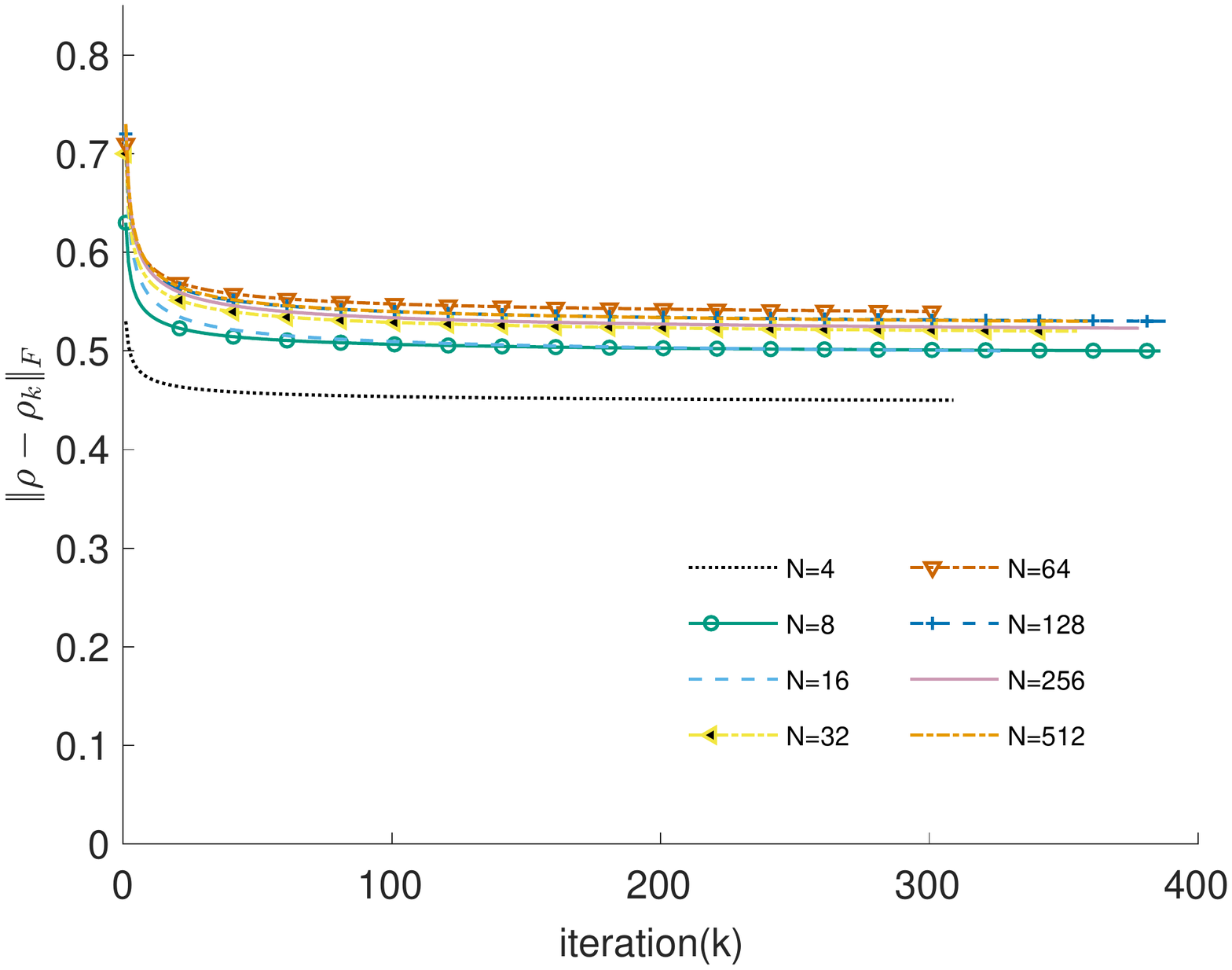}
       \caption{$\norm{\rho-\rho_{k}}_{F}$ in Example~\ref{ex3} }
    \end{figure}
  \end{minipage}
 \begin{minipage}[t]{0.5\linewidth}
    \begin{figure}[H]
      \centering
      \includegraphics[width=0.9\linewidth]{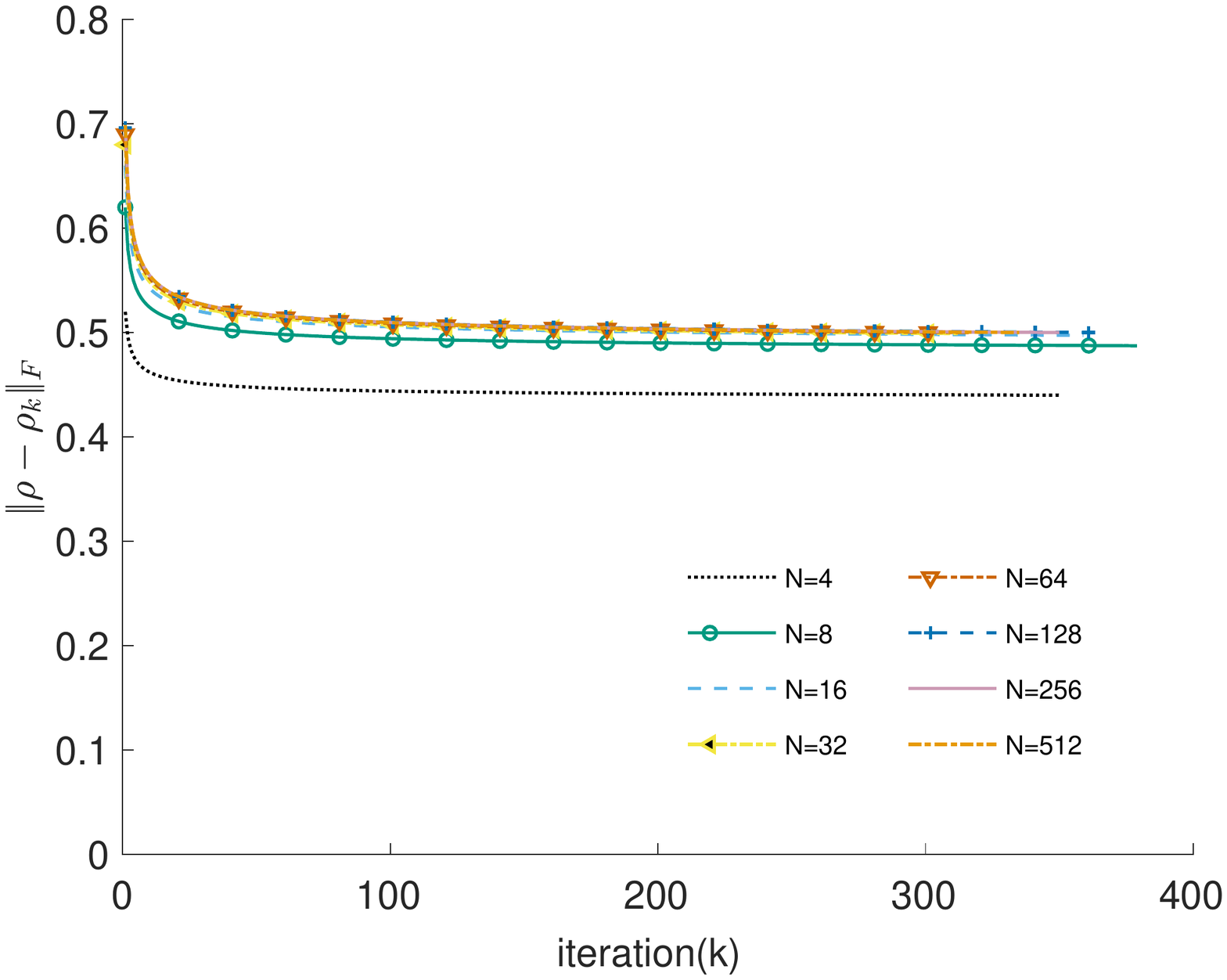}
        \caption{$\norm{\rho-\rho_{k}}_{F}$ in Example~\ref{ex4} }
    \end{figure}
  \end{minipage}

\section{Improvement of the Algorithm~\ref{mainalgorithm}}
\label{sec:improve}
 In fact, we can accelerate the convergence by solving the following quadratic programming
  \begin{equation}
    \label{S-sep-v7:eq:3}
    \begin{array}[t]{rl}
      \max\limits_{\lambda_i}& \norm{\rho-\sum\limits_{i=0}^{k-1}\lambda_{i}\ket{x_{i},x_{i}}\bra{x_{i},x_{i}}}_{F}\\
      \text{ s.t.}  & \lambda_i\geqslant 0,i=0,1,\ldots,k-1,
    \end{array}
  \end{equation}
  at $k$-th step to update the coefficient of $\rho_k$, which is closer to $\rho$.
  
    Note that the Algorithm~\ref{mainalgorithm} may not continue if a small local maximum of $\langle
    x,x|\rho-\rho_{k}|x,x\rangle$ on the unit sphere is found
    such that
    \begin{equation}
      \label{S-sep-v7:eq:4}
      \langle \rho-\rho_k, \ket{x,x}\bra{x,x}-\rho_{k}\rangle <0.
    \end{equation}
    Therefore, global optimization algorithm is need to guarantee the convergence of Algorithm~\ref{mainalgorithm}. For example,
    the multi-start strategy can be used to have better results then just using Algorithm~\ref{combineNewton}.

    In order to be specific, we consider in this paper mostly the real matrices, however, our algorithm can be easily
    generalized to real case. It turns out that the states in complex case is the bosonic states, that is the states in
    the  system of indistinguishable particles~\cite{Eckert_2002}. For a bosonic state $\rho$, it is separable if and only it
    can be written as
    \begin{equation}
      \label{S-sep-v8:eq:1}
     \rho = \sum_i\lambda_i\ket{x_i,x_i}\bra{x_i,x_i}
   \end{equation}
   where $\ket{x_i}\in\complex^{N}$. Algorithm~\ref{mainalgorithm} can be generalized to the complex case only if we
   modify the transposition of real operators to Hermitian conjugation of complex
   operators. Algorithm~\ref{combineNewton} can also be generalized to the complex case, that it find the maximum value
   of
   \begin{equation}
     \label{S-sep-v8:eq:2}
     \bra{x_i,x_i}\sigma \ket{x_i,x_i},
   \end{equation}
   with $\sigma $ being a matrices in the indistinguishable system. We can utilize the Wirtinger calculus~\cite{wirtinger1927formalen} to compute the
   gradient vectors  and Hessian matrices. The Algorithm~\ref{combineNewton} thus can be modified accordingly.
   
   It should also be noted that our algorithm can be extended easily to multipartite case, where we need only calculate the gradient vector and Hessian maxtrix and modify the Algorithm~\ref{combineNewton} and \ref{mainalgorithm} respectively. In order to be concise, we omit describing the implementation details here. 
    
\section{Conclusion}
   \label{remarks}
In this paper, we consider a subclass of quantum states in the $N\otimes N$ space, namely, the completely symmetric  
 states, which is a similar conception of “supersymmetric tensor” originally in the filed of tensor decomposition.
 Inspired by the special structure, we conjecture that the completely symmetric state is separable if and only if it is S-separable, which is proved to be true when the rank is less than $4$ or $N+1$. However, it need further research to check whether the validness of this conjecture.

 Besides, we propose a numerical algorithm which is able to detect such property with the convergence behavior analyzed. The numerical results show our algorithm is efficient. Furthermore, our algorithm is suitable to check the separability of the Bosonic states. Our algorithm is hopeful to be  a aided tool to the study of entanglement in the symmetric system.
 
 In the future work,  Conjecture~\ref{conj2} should be investigated in more complicate cases, for example, the states are
 supported in the higher dimensional space or
 of higher ranks.  It is also of interest to consider this problem in the multipartite system. Another interesting
 problem is the consider which kind of  states can be transformed to the  completely symmetric states by the invertible
 local operator. The Example~\ref{ex0} and~\ref{ex01} are two kind states which  satisfy this condition.

\appendix
\section{Proof of  Eq. ~\eqref{goesqp:eq:43}}%
\label{apdeq1}
In this appendix, we prove Eq.~\eqref{goesqp:eq:43}. Let us describe this question formally with a lemma.
\begin{lem}
  Let $x_{*}$ and $x_{\bot}$ are two orthogonal unit vectors in the $\real^{N}$ space and 

 \begin{equation}
   x = \cos(\theta)x_{*}+ \sin(\theta)x_{\bot},\theta\in[0,2\pi].
 \end{equation}
 Then we have
  \begin{equation}
   \abs{ \sin(\theta)} = \sqrt{\norm{x-x_{*}}^{2} - \frac{\norm{x-x_{*}}^4}{4}}.
  \end{equation}
\end{lem}
\begin{proof}
  The following graph shows the relations of $x$,$x_{*}$, and $x_{\bot}$ when $\langle
  x,x_{*}\rangle >0$ (left one) and  $\langle x,x_{*}\rangle <0$ (right one):      
  \begin{figure}[H]
    \centering
      \includegraphics[width=0.30\linewidth]{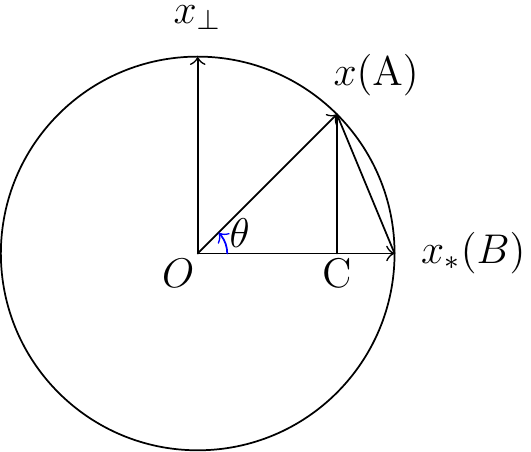}\qquad
      \includegraphics[width=0.30\linewidth]{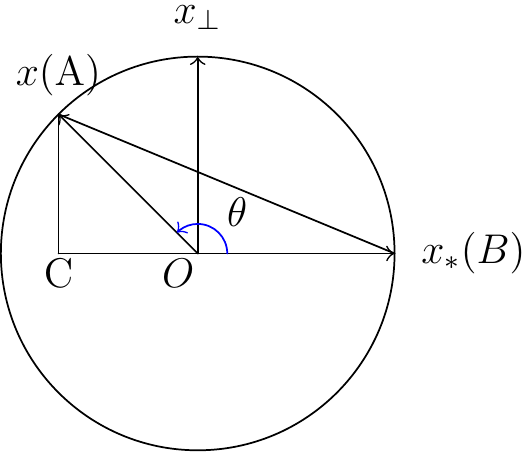}         
        \caption{$ x = \cos(\theta)x_{*}+ \sin(\theta)x_{\bot}$}
  \end{figure}
  From the above graphs, we have
  \begin{equation}
    \label{S-sep-v5.1:eq:2}
    \norm{x-x_{*}} = \abs{AB}.
  \end{equation}
  And
  \begin{equation}
    \label{S-sep-v5.1:eq:3}
    \begin{split}
      \abs{OC} &= \abs{\cos(\theta)},\\
      \abs{BC}& = 1 - \cos(\theta),\\
      \abs{AC}& = \sin(\theta).
    \end{split}
    \end{equation}
    Moreover,
    \begin{equation}
      \label{S-sep-v5.1:eq:4}
      \begin{split}
        \abs{AB}& = \sqrt{\abs{AC}^2+\abs{BC}^2}\\
        & = \sqrt{(1-\cos(\theta))^2+\sin^2(\theta)}\\
        & = \sqrt{2-2\cos(\theta)}.
      \end{split}
    \end{equation}
    Therefore,
    \begin{equation}
      \label{S-sep-v5.1:eq:5}
      \cos(\theta) = 1 - \frac{\abs{AB}^2}{2}.
    \end{equation}
    Forward,
    \begin{equation}
      \label{S-sep-v5.1:eq:6}
      \begin{split}
        \abs{  \sin(\theta)} & = \sqrt{1-\cos^2(\theta)}\\
        & = \sqrt{1 - (1-\frac{\abs{AB}^2}{2})^2}\\
        & = \sqrt{\abs{AB}^2-\frac{\abs{AB}^2}{4}}\\
        & = \sqrt{\norm{x-x_{*}}^{2} - \frac{\norm{x-x_{*}}^4}{4}}.
      \end{split}
    \end{equation}
  \end{proof}
  \section{Proof of Eq.~\eqref{xsqp}}
  \label{apdxsqp}
  In this appendix, we prove Figure~\ref{xsqp}, that is to prove
  \begin{equation}
    \label{S-sep-v5.1:eq:7}
    \norm{x_{\sqp}-x_{*}}\leqslant \norm{x_{\nt}-x_*},
  \end{equation}
  where $\norm{x_{\sqp}}\leqslant \norm{x_{\nt}}$ and $x_{*},x_{\sqp}$  are unit vectors.
  The following Figure~\ref{fig:xsqp} shows the relationship of $x_{\sqp}$ and $x_{\nt}$.
  \newline
  \begin{figure}[H]%
    \centering
      \includegraphics[width=0.25\linewidth]{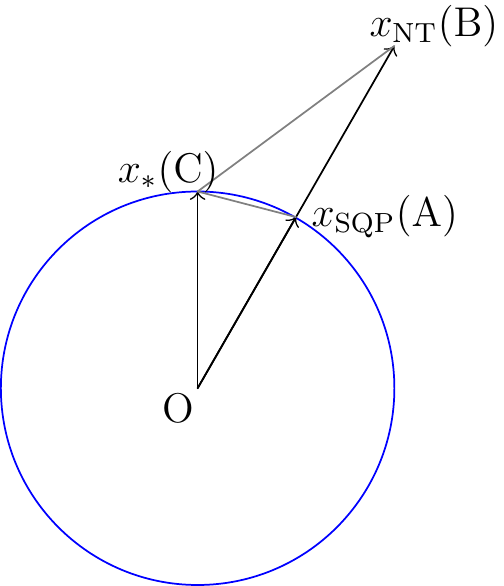}
      \caption{Relation between $x_{\sqp}$ and $x_{\nt}$}
         \label{fig:xsqp}
       \end{figure}
       Note that
       \begin{equation}
         \label{S-sep-v5.1:eq:8}
         \angle OAC = \frac{1}{2}(\pi - \angle AOC).
       \end{equation}
       
       Hence,
       \begin{equation}
         \label{S-sep-v5.1:eq:9}
         \begin{split}
           \angle BAC &= \pi -\angle OAC\\
           & = \pi - \frac{1}{2}(\pi - \angle AOC)\\
           & = \frac{\pi}{2}+\frac{1}{2}\angle AOC\\
           & > \frac{\pi}{2}.
           \end{split}
         \end{equation}
         Moreover,
         \begin{equation}
           \label{S-sep-v5.1:eq:10}
           \sin(\angle ABC) \leqslant \sin (\pi - \angle BAC) = \sin(\angle BAC).
         \end{equation}
         Forward,
         \begin{equation}
           \label{S-sep-v5.1:eq:11}
           \frac{\norm{x_{\sqp}-x_{*}}}{\norm{x_{NT}-x_{*}}}= \frac{\sin(\angle ABC)}{\sin(\angle BAC)}.
         \end{equation}
         By Eq.~\ref{S-sep-v5.1:eq:10}, we have
         \begin{equation}
           \label{S-sep-v5.1:eq:12}
             \frac{\norm{x_{\sqp}-x_{*}}}{\norm{x_{NT}-x_{*}}}\leqslant 1,
           \end{equation}
           which completes our proof.

\end{document}